%% file: main.tex
\documentclass[runningheads]{llncs}

\pdfoutput=1

\input{layout}
\input{macros}

\begin{document}

\input{title}

\input{intro}

\input{tessla}

\input{abstractions}

\input{compositional}

\input{queue}

\input{empirical}

\input{conclusion}

\bibliographystyle{splncs04}
\bibliography{references}

\input{appendix}

\end{document}

%% file: layout.tex
\usepackage{amsmath}

\usepackage[utf8]{inputenc}
\usepackage[T1]{fontenc}

\usepackage{graphicx}

\usepackage{mathtools}

\usepackage{microtype}

\usepackage{paralist}

\usepackage[para]{footmisc}

\usepackage{hyperref}
\hypersetup{breaklinks=true,
            pdfborder={0 0 0},
            pdfhighlight={/N}}
\usepackage[all]{hypcap}

\usepackage[table]{xcolor}

\usepackage{listings}

\usepackage{tikz}
\usetikzlibrary{shapes.misc,calc,decorations,decorations.pathreplacing,backgrounds,shapes.multipart,fit,positioning}

\usepackage{xspace}

\usepackage{utf8-math}

\usepackage{streampicture}

\usetikzlibrary{arrows.meta}
\tikzset{
	event/.style={
		draw,
		inner sep=.5pt,
		circle,
		minimum width=6pt
	},
	events/.style={
		font=\tiny,
		xscale=.5,
		yscale=.3
	},
	unit/.style={
		event,
		path picture={ 
			\draw[black]
			(path picture bounding box.south east) -- (path picture bounding box.north west)
			(path picture bounding box.south west) -- (path picture bounding box.north east);
		}
	}
}

\usepackage{wrapfig}

%% file: macros.tex
\let\op\operatorname

\newcommand{\tessla}{TeSSLa\xspace}

\let\phi\varphi
\let\epsilon\varepsilon

\DeclareMathOperator{\ig}{\textit{ig}}

\DeclareMathOperator{\gap}{\smallsmile}

\newcommand{\unitsym}{\boxempty}

\newcommand{\testf}{\textrm{t\kern-0.14em f}}
\newcommand{\testrue}{\textrm{t\kern-0.16em t}}
\newcommand{\tesfalse}{\textrm{f\kern-0.1em f}}

\newcommand{\KW}[1]{\ensuremath{\mathbf{#1}}}
\newcommand{\NIL}{\KW{nil}\xspace}
\newcommand{\UNIT}{\KW{unit}\xspace}
\newcommand{\TIME}{\KW{time}\xspace}
\newcommand{\LIFT}{\KW{lift}\xspace}
\newcommand{\LAST}{\KW{last}\xspace}
\newcommand{\DELAY}{\KW{delay}\xspace}
\newcommand{\KWabs}[1]{\ensuremath{\KW{#1}^\#}}
\newcommand{\NILabs}{\KWabs{nil}\xspace}
\newcommand{\UNITabs}{\KWabs{unit}\xspace}
\newcommand{\TIMEabs}{\KWabs{time}\xspace}
\newcommand{\LIFTabs}{\KWabs{lift}\xspace}
\newcommand{\LASTabs}{\KWabs{last}\xspace}
\newcommand{\LASTbot}{\ensuremath{\KW{last}^\#_\bot}\xspace}
\newcommand{\LASTgap}{\ensuremath{\KW{last}^\#_{\gap}}\xspace}
\newcommand{\DELAYabs}{\KWabs{delay}\xspace}
\newcommand{\DELAYbot}{\ensuremath{\KW{delay}^\#_\bot}\xspace}
\newcommand{\DELAYgap}{\ensuremath{\KW{delay}^\#_{\gap}}\xspace}

\newcommand{\LASTTIME}{\KW{lastTime}}
\newcommand{\LASTTIMEabs}{\ensuremath{\KW{lastTime}^{\#}}}

\newcommand{\DEF}[1]{\textsf{#1}}
\newcommand{\AUX}[1]{\ensuremath{\mathit{#1}}}

\newcommand{\Ticks}{\AUX{ticks}}

\DeclareMathOperator{\DefinedAs}{\stackrel{\text{def}}{=}}

\newcommand{\KWbullet}{\noindent\hskip-7mm\makebox[7mm]{\hfill\tikz[baseline]\draw[-{Triangle[length=4pt]}](0,.6ex)--++(4pt,0);\kern2pt}}
\let\KWabsbullet\KWbullet
\newcommand{\KWlightbullet}{\noindent\hskip-7mm\makebox[7mm]{\hfill\tikz[baseline]\draw[-{Triangle[open,length=4pt]}](0,.6ex)--++(4pt,0);\kern2pt}}

\newcommand{\mainrowcolors}{\rowcolors{1}{black!10}{black!3}}
\newenvironment{zebratabular}{\mainrowcolors\begin{tabular}}{\end{tabular}}
\newcommand{\setrownumber}[1]{\global\rownum#1\relax}
\newcommand{\headerrow}{\rowcolor{black!20}\setrownumber1}

%% file: title.tex

\newcommand{\Thanks}{\thanks{This work was funded in part by the
    Madrid Regional Government under project \emph{S2018/TCS-4339
    (BLOQUES-CM)}, by EU H2020 projects 731535 \emph{Elastest} and 732016 \emph{COEMS}, by
    Spanish National Project \emph{BOSCO (PGC2018-102210-B-100)} and by
    the BMBF project \emph{ARAMiS II} with funding~ID 01\,IS\,16025.}}

\title{Runtime Verification For Timed Event Streams With Partial Information\Thanks}
%
%
\author{Martin~Leucker\inst{1} \and
	César~Sánchez\inst{2} \and
	Torben~Scheffel\inst{1} \and
	Malte~Schmitz\inst{1} \and
	Daniel~Thoma\inst{1}}
\authorrunning{M. Leucker, C. Sánchez, T. Scheffel, M. Schmitz, and D. Thoma}
\institute{University of Lübeck, Germany
	\email{\{leucker,scheffel,schmitz,thoma\}@isp.uni-luebeck.de}\\
   \and
	IMDEA Software Institute, Spain\\
	\email{cesar.sanchez@imdea.org}}
\maketitle

\begin{abstract}
  Runtime Verification (RV) studies how to analyze execution traces of a system under observation.
  Stream Runtime Verification (SRV) applies stream transformations to obtain information from observed traces.
  Incomplete traces with information missing in gaps pose a common challenge when applying RV and SRV techniques to real-world systems as RV approaches typically require the complete trace without missing parts.
  This paper presents a solution to perform SRV on incomplete traces based on abstraction.
  We use TeSSLa as specification language for non-synchronized timed event streams and define abstract event streams representing the set of all possible traces that could have occurred during gaps in the input trace.
  We show how to translate a TeSSLa specification to its abstract counterpart that can propagate gaps through the transformation of the input streams and thus generate sound outputs even if the input streams contain gaps and events with imprecise values.
  The solution has been implemented as a set of macros for the original TeSSLa and an empirical evaluation shows the feasibility of the approach.
\end{abstract}

%% file: intro.tex

\section{Introduction}
\label{sec:intro}

Runtime verification (RV) is a dynamic formal method for software
system reliability.
RV studies how to analyze and evaluate traces against formal specifications
and how to obtain program traces from the system under observation,
e.g., through software instrumentation or utilization of processors'
embedded trace units.
Since RV only inspects one execution trace of the system, it is often
regarded to be a readily applicable but incomplete approach, that
combines formal verification with testing and debugging.

Most early RV languages were based on logics common in static
verification, like LTL~\cite{manna95temporal}, past LTL adapted for
finite
paths~\cite{havelund02synthesizing,eisner03reasoning,bauer11runtime},
regular expressions~\cite{sen03generating} or timed regular
expressions~\cite{TRE}.
For these logics, the monitoring problem consists on computing a
Boolean verdict indicating whether the trace fulfills the
specification.
In contrast to static analysis, however, considering only a single
concrete trace enables the application of more complex analyses:
Stream Runtime Verification (SRV)~\cite{dangelo05lola,bozzelli14foundations,tessla}
uses stream transformations
to derive additional streams as verdicts from the input streams.
Using SRV one can still check if the input stream is conformant with a
specification, but additionally verify streams in terms of their
events' data:
streams in SRV can store data from richer domains than Booleans,
including numerical values or user defined data-types, so SRV
languages can extract quantitative values and express quantitative
properties like ``\emph{compute the average retransmission time}'' or
``\emph{compute the longest duration of a function}''.
SRV cleanly separates the temporal dependencies that the stream
transformation algorithms follow from the concrete operations to be
performed on the data, which are specific to each data-type.
As an example for SRV consider the trace diagram on the left of
Fig.~\ref{fig:example}.
We consider non-synchronized event streams, i.e., sequences of events
with increasing timestamps and values from a data domain.
Using non-synchronized event streams one can represent events arriving on different
streams with different frequencies in a compact way with little computation
overhead because there is no need to process additional synchronization events
in the stream-transformation process.
In this paper we use the \tessla specification language
\cite{tessla}, an SRV language for non-synchronized, timed event
streams.
TeSSLa has been defined to be general enough to allow for a natural
translation from other common SRV formalisms, e.g., Lola~\cite{dangelo05lola} and
Striver~\cite{gorostiaga18striver}.
Therefore, our results carry over to these languages as well.

\begin{figure}[t]
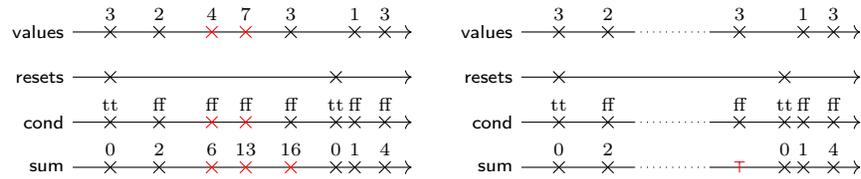

\centering
\begin{streampicture}[yscale=1.2]
  \begin{stream}{\DEF{values}}
    \definedto{9}
    \eventat[val=3]{1}
    \eventat[val=2]{2.3}
    \eventat[val=4,red]{3.7}
    \eventat[val=7,red]{4.6}
    \eventat[val=3]{5.8}
    \eventat[val=1]{7.5}
    \eventat[val=3]{8.3}
  \end{stream}
  \begin{stream}{\DEF{resets}}
    \definedto{9}
    \eventat{1}
    \eventat{7}
  \end{stream}
  \begin{stream}{\DEF{cond}}
    \definedto{9}
    \eventat[val=tt]{1}
    \eventat[val=ff]{2.3}
    \eventat[val=ff,red]{3.7}
    \eventat[val=ff,red]{4.6}
    \eventat[val=ff]{5.8}
    \eventat[val=tt]{7}
    \eventat[val=ff]{7.5}
    \eventat[val=ff]{8.3}
  \end{stream}
  \begin{stream}{\DEF{sum}}
    \definedto{9}
    \eventat[val=0]{1}
    \eventat[val=2]{2.3}
    \eventat[val=6,red]{3.7}
    \eventat[val=13,red]{4.6}
    \eventat[val=16,red]{5.8}
    \eventat[val=0]{7}
    \eventat[val=1]{7.5}
    \eventat[val=4]{8.3}
  \end{stream}
\end{streampicture}
\quad
\begin{streampicture}[yscale=1.2]
  \begin{stream}{\DEF{values}}
    \definedto{3}
    \eventat[val=3]{1}
    \eventat[val=2]{2.3}
    \undefinedto{5}
    \definedto{9}
    \eventat[val=3]{5.8}
    \eventat[val=1]{7.5}
    \eventat[val=3]{8.3}
  \end{stream}
  \begin{stream}{\DEF{resets}}
    \definedto{9}
    \eventat{1}
    \eventat{7}
  \end{stream}
  \begin{stream}{\DEF{cond}}
    \definedto{3}
    \eventat[val=tt]{1}
    \eventat[val=ff]{2.3}
    \undefinedto{5}
    \definedto{9}
    \eventat[val=ff]{5.8}
    \eventat[val=tt]{7}
    \eventat[val=ff]{7.5}
    \eventat[val=ff]{8.3}
  \end{stream}
  \begin{stream}{\DEF{sum}}
    \definedto{3}
    \eventat[val=0]{1}
    \eventat[val=2]{2.3}
    \undefinedto{5}
    \definedto{9}
    \topat[red]{5.8}
    \eventat[val=0]{7}
    \eventat[val=1]{7.5}
    \eventat[val=4]{8.3}
  \end{stream}
\end{streampicture}
\caption{Example trace for a typical SRV specification (left) with two input streams
\DEF{values} (with numeric values) and \DEF{resets} (with no internal value).
The  intention of the specification  is to accumulate in the output stream \DEF{sum} all
values since the last reset. 
The intermediate stream \DEF{cond} is derived from the input streams indicating
if \DEF{reset} has currently the most recent event, and thus the sum should be reset to 0.
If the input streams contain gaps (dotted regions on the right)
some information
can no longer be computed, but after a reset event the computation recovers from the
data loss during the gap. $\top$ denotes events with unknown data.}
\label{fig:example}
\end{figure}

Since RV is performed on traces obtained from the system under test in the
deployed environment, it is a common practical problem for RV techniques
that the traces do not cover the entire run of the system.
However, most of the previous RV approaches require the trace to be available
without any interruptions in order to obtain a verdict, because
this knowledge is assumed in the semantics of the specification logics.
Especially in the case of interrupted traces with some data losses
applying previous RV techniques can be very challenging.
Unfortunately those traces occur very often in practical
testing and debugging scenarios, e.g., due to interrupted experiments,
buffer overflows, network errors or any other temporary problem
with the trace retrieval.

In this paper we present a solution to the problem of evaluating
traces with imprecise values and even interrupted traces.
Our only assumption is that we have exact knowledge of the imprecision
of the trace in the following sense: (1) for events with imprecise
values we know the range of values and (2) for data losses we know
when we stop getting information and when the trace becomes reliable
again. We call such a sequence of uncertainty a \emph{gap} in the
trace.
Our solution automatically propagates gaps and imprecisions, and
allows to obtain sound verdicts even in the case of missing information
in the input trace.
The assumption is reasonable in our target application: online
non-intrusive monitoring low-level embedded software.

Fig.~\ref{fig:example} on the right displays a case where the input
stream \DEF{values} has a long gap in the middle.
It is not possible to determine the events in the output stream
\DEF{sum} during that gap, because we do not even know if and how many
events might have happened during that gap.
Thus, the intermediate stream \DEF{cond} and the output stream
\DEF{sum} simply copy that gap representing any possible combination
of events that might occur.
The first event after the gap is the one with the value 3 on
\DEF{values}.
Because no reset happened after the end of the gap, we would add 3 to
the latest event's value on \DEF{sum}, but the gap is the latest on
\DEF{sum}.
Thus, we only know that this input event on \DEF{values} causes an
event on \DEF{sum} independently of what might have happened during
the gap, but the value of that event completely depends on possible
events occurring during the gap.
After the next event on \DEF{reset} the values of the following events on \DEF{sum} are independent of any previous events.
The monitor can fully recover from the missing information during the gap and can again produce events with precise values.

In order to realize this propagation of gaps through all the steps of
the stream-transformation we need to represent all potentially
infinitely many concrete traces (time is dense and values are for
arbitrary domains) that might have happened during gaps and imprecise
events.
An intuitive approach would be a symbolic representation in terms of constraint formulas to describe the set of all possible streams.
These formulas would then be updated while evaluating the input trace.
While such a symbolic execution might work for shorter traces, the representation can grow quickly with each input event.
Consequently the computational cost could grow prohibitively with the trace length for many input traces.
Instead, in this paper we introduce a framework based on
abstraction~\cite{cousot77abstract,cousot92abstract}.
We use abstraction in two ways:
\begin{compactenum}[(1)]
\item Streams are lifted from concrete domains of data to abstract
domains to model possible sets of values.
For example, in our solution a stream can store intervals as abstract
numerical values.
\item We define the notion of abstract traces, which extend
timed streams with the capabilities of representing gaps.
Intuitively, an abstract trace over-approximates the sets of concrete
traces that can be obtained by filling the gaps with all possible
concrete events.
\end{compactenum}

\noindent
Our approach allows for both gaps in the input streams as well as
events carrying imprecise values. Such imprecise values can be
modelled by abstract domains, e.g., intervals of real numbers.
Since we rely on abstraction, we can avoid false negatives and false
positives in the usual sense:
concrete verdicts are guaranteed to hold and imprecise verdicts are
clearly distinguished from concrete verdicts.
The achievable precision depends on the specification and the input trace.

After reproducing the semantics of the basic TeSSLa operators in Section~\ref{sec:tessla}, we introduce abstract semantics of the existing basic operators of \tessla in Section~\ref{sec:abstraction}.
Using these abstract TeSSLa operators, we can take a TeSSLa specification on streams and replace every TeSSLa operator with its abstract counterpart and derive an abstraction of the specification on abstract event streams.
We show that the abstract specification is a sound abstraction of the
concrete specification, i.e., every concrete verdict generated by the
original specification on a set $S$ of possible input traces is
represented by the abstract verdict applied to an abstraction of $S$.
We further show that the abstract TeSSLa operators are a perfect
abstraction of their concrete counterparts, i.e., that applying the
concrete operator on all individual elements of $S$ doesn't get you
more accurate results.
Finally, we show that an abstract TeSSLa specification can be implemented using the existing TeSSLa basic operators by representing an abstract event stream as multiple concrete event streams carrying information about the events and the gaps.
Since the perfect accuracy of the individual abstract TeSSLa operators does not guarantee perfect accuracy of their compositions, we discuss the accuracy of composed abstract TeSSLa specifications in Section~\ref{sec:compositional}.
Next we present in Section~\ref{sec:queue} an advanced use-case where
we apply abstract TeSSLa to streams over a complex data domain of
unbounded queues, which are used to compute the average of all events
that happened in the sliding window of the last five time units.
In the final Section~\ref{sec:evaluation} we evaluate the overhead and the accuracy of the abstractions presented in this paper on representative example specifications and corresponding input traces with gaps.

\paragraph{Related Work.}
SRV was pioneered by LOLA~\cite{dangelo05lola,LOLA2,rtlola}.
TeSSLa~\cite{tessla} generalises to asynchronous streams the original
idea of LOLA of recursive equations over stream transformations.  Its
design is influenced by formalisms like stream programming languages
\cite{halbwachs87lustre,berry00foundations,gautier87signal} and
functional reactive programming~\cite{eliot97functional}.
Other approaches to handle data and time constraints include
Quantitative Regular Expressions QRE~\cite{QRE} and Signal Temporal
Logic~\cite{STL}.

While ubiquitous in practice, the problem of gaps in an observation
trace has not been studied extensively.
To the best of our knowledge, abstraction techniques have not been
applied to the evaluation of stream-based specifications.
However, approaches to handle the absence of events or ordering
information have been presented for
MTL~\cite{DBLP:conf/fsttcs/BasinKZ15} and past-time
LTL~\cite{DBLP:conf/rv/WangASL11}.
State estimation based on Markov models has been applied to replace
absent information by a probabilistic
estimation~\cite{DBLP:conf/rv/StollerBSGHSZ11}.
The concept of abstract interpretation used throughout this paper has
been introduced in \cite{DBLP:conf/popl/CousotC77}.


%% file: tessla.tex

\section{The TeSSLa Specification Language}
\label{sec:tessla}

A \emph{time domain} is a totally ordered semi-ring $(𝕋,0,1,+,·,≤)$
that is positive, i.e., $∀_{t∈𝕋}\, 0 ≤ t$. 
We extend the order on time domains to the set $𝕋_∞ = 𝕋 ∪ ｛∞｝$ with
$∀_{t∈𝕋}\, t < ∞$.
Given a time domain $𝕋$, an \emph{event stream} over a data
domain $𝔻$ is a finite or infinite sequence $s = t_0 d_0 t_1 \dots ∈ 𝓢_𝔻 =
(𝕋·𝔻)^ω ∪ (𝕋·𝔻)^*·(𝕋_∞ ∪ 𝕋·𝔻_⊥)$ where
$\mathbb{D}_\bot := \mathbb{D} \cup \{\bot\}$ and $t_i <
t_{i+1}$ for all $i$ with $0<i+1<|s|$ ($|s|$ is $∞$ for infinite
sequences).
An infinite event stream is an infinite sequence of timestamps and data values
representing the stream's events.
A finite event stream is a finite sequence of timestamped events up to
a certain timestamp that indicates the progress of the stream.
A stream can end with:
\begin{compactitem}[--]
\item a timestamp without a data value that denotes progress up
  to but not including that timestamp,
\item a timestamp followed by $\bot$ (or a data value) which denotes
  progress up to and including that timestamp (and an event at that timestamp),
\item $∞$, which indicates that no additional events will ever arrive on this stream.
\end{compactitem}
We refer to these cases as \emph{exclusive}, \emph{inclusive} and
\emph{infinite progress}, resp.

\noindent
Streams $s ∈ 𝓢_𝔻$ can be seen as functions $s: 𝕋→𝔻∪｛⊥,\mathrm{?}｝$
such that $s(t)$ is a value $d$ if $s$ has an event with value $d$
at time $t$ or $⊥$ if there is no event at time $t$. For timestamps
after the progress of the stream $s(t)$ is ?.
Formally, $s(t) = d$ if $s$ contains $t d$, $s(t) = ⊥$ if $s$ does not
contain $t$, but contains a $t' > t$ or $s$ ends in $t⊥$, and $s(t) =
\mathrm{?}$ otherwise.
We use $\Ticks(s)$ for the set $\{t \in \mathbb T \mid s(t) \in \mathbb D\}$ of timestamps where $s$ has events.
A stream $s$ is a \emph{prefix} of stream $r$ if $\forall_{t∈𝕋}
s(t) ∈｛r(t),\mathrm{?}｝$.
We use the unit type $𝕌=｛\unitsym｝$ for streams carrying only
the single value $\unitsym$.

A TeSSLa specification consists of a collection of stream variables
and possibly recursive equations over these variables using the
operators $\KW{nil}$, $\KW{unit}$, $\KW{time}$, $\KW{lift}$,
$\KW{last}$ and $\KW{delay}$. 
The semantics of recursive equations is given as the least fixed-point of
the equations seen as a function of the stream variables and
fixed input streams.
See \cite{tessla} for more details
and Appendix~\ref{app:semantics-example} for an elaborated example.

\KWbullet
$\KW{nil} = ∞ \in 𝓢_{\emptyset}$ is the stream without any events and infinite progress.

\KWbullet
$\KW{unit} = 0 \unitsym ∞ \in 𝓢_𝕌$ is the stream with a
single unit event at timestamp zero and infinite progress.

\KWbullet
$\KW{time}: 𝓢_𝔻 → 𝓢_𝕋, \KW{time}(s) := z$ maps the event's values to their timestamps:\\
$z(t) = t$ if $t\in\Ticks(s)$ and $z(t) = s(t)$ otherwise.

\KWbullet
$\KW{lift}: ({𝔻_1}_\bot × \dots × {\mathbb{D}_n}_\bot \to 𝔻_\bot) → (𝓢_{𝔻_1}×
\dots × 𝓢_{𝔻_n}→𝓢_𝔻), \KW{lift}(f)(s_1, \ldots, s_n) := z$
lifts a function $f$ on values to a function on streams by applying $f$
to the stream's values for every timestamp. 
The function $f$ must not generate new events, i.e., must fulfill
$f(\bot,\dots,\bot) = \bot$.
\[
  z(t)=\begin{cases}
    f(s_1(t), \dots, s_n(t)) & \text{if $s_1(t)\neq\mathrm{?},\dots, s_n(t)\neq\mathrm{?}$} \\
    ? & \text{otherwise}
  \end{cases}
\]

\KWbullet
$\KW{last}: 𝓢_𝔻 × 𝓢_{𝔻'} → 𝓢_𝔻, \KW{last}(v, r) := z$ takes a stream $v$ of values and a stream $r$ of triggers.
It outputs an event with the previous value on $v$ for every event on $r$.
\[
  z(t) =
  \begin{cases}
    d & t\in\Ticks(r) \text{ and }\exists_{t' < t}\AUX{isLast}(t, t', v, d)\\
    \bot & r(t) = \bot \text{ and } \AUX{defined}(z, t) \text{, or  } \forall_{t'<t} v(t') = \bot \\
      ? & \text{ otherwise}
  \end{cases}
\]
$\AUX{isLast}(t, t', v, d) \DefinedAs v(t')=d \wedge ∀_{t'' | t' < t'' < t} v(t'') = ⊥$ holds
if $t'd$ is the last event on $v$ until $t$, and
$\AUX{defined}(z, t) \DefinedAs \forall_{t' < t} z(t') ≠ \mathrm{?}$ holds if $z$ is defined
until $t$ (exclusive).

Using the basic operators we can now derive the following
utility functions:

\KWlightbullet
$\DEF{const}(c)(a) := \KW{lift}(f_c)(a)$ with $f_c(d) := c$. 
This function maps the values of all events of the input stream $a$ to
a constant value $c$.
Using $\DEF{const}$ we can lift constants into streams representing a
constant signal with this value, e.g.,
$\DEF{true} := \DEF{const}(\DEF{true})(\KW{unit})$ or
$\DEF{zero} := \DEF{const}(0)(\KW{unit})$.

\KWlightbullet
$\DEF{merge}(x, y) := \KW{lift}(f)(x, y)$ with $f(a≠⊥, b) = a$
  and $f(⊥, b) = b$, which combines events from two streams,
  prioritizing the first stream.

Event streams in \tessla can also be interpreted as a \emph{continuous
  signals}.
Using \KW{last} one can query the last known value of an event stream
$s$ and interpret the events on $s$ as points where a piece-wise
constant signal changes its value.  
By combining the $\KW{last}$ and $\KW{lift}$ operators, we can realize:

\KWlightbullet
\emph{signal lift} for total functions $f: \mathbb D \times
\mathbb D' \to \mathbb D''$ as $\DEF{slift}(f)(x, y) :=
\KW{lift}(g_f)(x',y')$ with $x' := \DEF{merge}(x, \KW{last}(x, y))$
and $y' := \DEF{merge}(y, \KW{last}(y, x))$, as well as $g_f(a≠⊥,
b≠⊥) := f(a, b)$, $g_f(⊥, b) := ⊥$, and $g_f(a, ⊥) := ⊥$.

\begin{example}
  \label{running-example}
  We can now specify the stream
  transformations shown on the left in \autoref{fig:example}
  in \tessla.  
  Let $\DEF{resets} \in \mathcal S_{\mathbb U}$ and $\DEF{values} \in
  \mathcal S_{\mathbb Z}$ be two external input event streams.  
  We then derive $\DEF{cond} \in \mathcal S_{\mathbb B}$ and
  $\DEF{lst}, \DEF{sum} \in \mathcal S_{\mathbb Z}$ as follows:
  
  \noindent
  \begin{minipage}{7cm}
    \begin{align*}
      \DEF{cond} &= \DEF{slift}(\le)(\KW{time}(\DEF{resets}), \KW{time}(\DEF{values})) \\
      \DEF{lst} &= \DEF{merge}(\KW{last}(\DEF{sum}, \DEF{values}), \DEF{zero}) \\
      \DEF{sum} &= \DEF{slift}(f)(\DEF{cond}, \DEF{lst}, \DEF{values})
    \end{align*}
  \end{minipage}
  \begin{minipage}{5cm}
    \begin{align*}
      & f: \mathbb B \times \mathbb Z \times \mathbb Z \to \mathbb Z \text{ with} \\
      & f(c,l,v) = \begin{cases}
        0 & \text{if } c = \mathrm{true} \\
        l + v & \text{otherwise}
      \end{cases}
    \end{align*}
  \end{minipage}
\end{example}

Using the operators described above one can only derive streams with timestamps
that are already present in the input streams.
To derive streams with events at computed timestamps one can use the
\KW{delay} operator, which is described in Appendix~\ref{app:delay}.


%% file: abstractions.tex

\section{Abstract \tessla}
\label{sec:abstraction}

\paragraph{Preliminaries.} Given two partial orders $(A,\preceq)$ and $(B, \preceq)$, a \emph{Galois
  Connection} is a pair of monotone functions $α: A → B$ and $γ: B →
A$ such that, for all $a ∈ A$ and $b ∈ B$, $α(a) \preceq b$ if and only if
$a \preceq γ(b)$.
Let $(A,\preceq)$ be a partial order, $f: A → A$ a monotone function and $γ:
B → A$ a function. 
The function $f^\#: B → B$ is an \emph{abstraction} of $f$
whenever, for all $b ∈ B$, $f(γ(b)) \preceq γ(f^\#(b))$.
If $(α, γ)$ is a Galois Connection between $A$ and $B$, the function
$f^\#: B → B$ such that $f^\#(b) := α(f(γ(b))$ is a
\emph{perfect abstraction} of $f$.

In this section we define the abstract counterparts of the \tessla
operators, listed in Section~\ref{sec:tessla}.
A \emph{data abstraction} of a data domain $𝔻$ is an abstract domain $𝔻^\#$
with an element $⊤ ∈ 𝔻^\#$ and an associated concretisation function
$γ: 𝔻^\# → 2^𝔻$ with $γ(⊤) = 𝔻$.
The abstract value $⊤$ represents any possible value from the data
domain and can be used to model an event with known timestamp but
unknown value. 
A \emph{gap} is a segment of an abstract event stream that represents
all combinations of events that could possibly occur in that segment
(both in terms of timestamps and values).
Hence an abstract event stream consists of an event stream over a data
abstraction and an associated set of known timestamps:

\begin{definition}[Abstract Event Stream]
  \label{def:pes}
  Given a time domain $𝕋$, an \emph{abstract event stream} over a data
  domain $𝔻$ is a pair $(s, Δ)$ with $s ∈ 𝓢_{𝔻^{\#}}$ and
  $Δ ⊆ 𝕋$ such that $Δ$ can be represented as union of intervals
  whose (inclusive or exclusive) boundaries are indicated by events in an event stream.
  Further, we require $s(t) \neq \bot \Rightarrow t \in Δ$.
  The set of all abstract event streams over $𝔻$ is denoted as $𝓟_𝔻$.
  The concretisation function $γ: 𝓟_𝔻 → 2^{𝓢_𝔻}$ is defined as
  \[γ((s,Δ)) = ｛s'｜∀_{t∈\Ticks(s)} s(t) ∈ γ(s'(t)) \wedge ∀_{t∈Δ\setminus\Ticks(s)} s(t) = s'(t) ｝\]
\end{definition}
	
\noindent
If the data abstraction is defined in terms of a Galois Connection a
refinement ordering and abstraction function can be obtained.
The refinement ordering $(𝓟_𝔻, \preceq)$ is defined as $(s_1, Δ_1) \preceq (s_2,
Δ_2)$ iff $Δ_1 ⊇ Δ_2$ and
$∀_{t∈\Ticks(s_2)} s_1(t) \preceq s_2(t) \wedge ∀_{t∈Δ_2\setminus\Ticks(s_2)} s_1(t) = s_2(t)$.
The abstraction function $α: 2^{𝓢_𝔻} → 𝓟_𝔻$ is defined as $α(S) =
\mathsf{sup}｛(s, 𝕋)｜s ∈ S｝$.
Note, if the data abstraction is defined in terms of a Galois
Connection, $(α, γ)$ is a Galois Connection between $2^{𝓢_𝔻}$ and
$𝓟_𝔻$.

An abstract event stream $s = (s',\Delta) \in 𝓟_𝔻$ can also be seen as a function
$s: \mathbb{T} \rightarrow \mathbb{D}^\# \cup \{\mathrm{?},\bot,\gap\}$ with
$s(t) = s'(t)$ if $t \in \Delta$ and $s(t) = \gap$ otherwise.
A particular point $t$ of an abstract event stream $s$ can be either
\begin{inparaenum}[(a)]
  \item directly at an event ($s(t) \in \mathbb D$), 
  \item in a gap ($s(t) = \gap$),
  \item in a gapless segment without an event at $t$ ($s(t) = \bot$), or
  \item after the known end of the stream ($s(t) = \mathrm{?}$).
\end{inparaenum}

We denote $\mathbb{D}^\#_\bot \DefinedAs \mathbb{D}^\# \cup \{\bot,
\gap\}$.  
If $\mathbb{D}^\#$ is a data abstraction of a data domain $\mathbb{D}$
with an associated concretisation function $\gamma$, then
$\mathbb{D}^\#_\bot$ is a data abstraction of $\mathbb{D}_\bot$ with
an associated concretisation function $\gamma_\bot: \mathbb{D}_\bot^\#
\to 2^{\mathbb{D} \cup \{\bot\}}$ with

\vskip1ex\par\noindent
\begin{minipage}{7cm}
\[
\gamma_\bot(d) = \begin{cases}
  \bot                   & \text{if } d=\bot \\
  \mathbb{D}\cup\{\bot\} & \text{if } d=\gap\\
  \gamma(d)              & \text{if } d\in\mathbb{D}^\#
  \end{cases}
\]
\end{minipage}
\begin{minipage}{4cm}
\begin{tikzpicture}[on grid, node distance=12mm and 6mm, inner sep=2pt]
  \node (abs true) {tt};
  \node[right=of abs true] (top) {$\top$};
  \node[right=of top] (abs false) {ff};
  \node[right=12mm of abs false, gray] (gap) {$\gap$};
  \node[right=of gap, gray] (abs bot) {$\bot$};
  \node[below=of abs true] (true) {tt};
  \node[below=of abs false] (false) {ff};
  \node[below=of abs bot, gray] (bot) {$\bot$};

  \path[->]
    (abs true) edge node[left] {$\gamma$} (true)
    (top.south) edge (true) edge (false)
    (abs false) edge (false)
    ($(gap.south) + (0,3pt)$) edge[gray] (true) edge[gray] (false) edge[gray] (bot)
    (abs bot) edge[gray] node[gray, right] {$\gamma_\bot$} (bot);

  \begin{pgfonlayer}{background}
    \node[fit=(abs true) (abs bot), fill=blue!10, rounded corners=3pt, inner sep=4pt] (B abs bot) {};
    \node[fit=(abs true) (abs false), fill=blue!25, rounded corners=3pt, inner sep=2pt] (B abs) {};
    \node[fit=(true) (bot), fill=red!10, rounded corners=3pt, inner sep=4pt] (B bot) {};
    \node[fit=(true) (false), fill=red!25, rounded corners=3pt, inner sep=2pt] (B) {};
  \end{pgfonlayer}

  \begin{scope}[inner sep=0pt]
    \node[right=1pt of B abs bot.north east,anchor=north west, gray] {$\mathbb B^\#_\bot$};
    \node[right=1pt of B abs.north east,anchor=north west] {$\mathbb B^\#$};
    \node[right=1pt of B bot.south east,anchor=south west, gray] {$\mathbb B_\bot$};
    \node[right=1pt of B.south east,anchor=south west] {$\mathbb B$};
  \end{scope}
\end{tikzpicture}
\end{minipage}
\vskip1ex\par
\noindent The above diagram shows a possible data abstraction
$\mathbb{B}^\#$ of $\mathbb{B}$ and the corresponding data
abstraction $\mathbb{B}^\#_\bot$.
Using the functional representation of an abstract event stream we can now define the abstract counterparts of the TeSSLa operators:

\KWabsbullet
$\NILabs=(\infty,\mathbb T) \in \mathcal{P}_\emptyset$ is the empty abstract stream without any gaps.

\KWabsbullet
$\UNITabs=(0\unitsym\infty, \mathbb T) \in \mathcal{P}_{\mathbb U}$ is the abstract stream without any gaps and a single event at timestamp 0.

\KWabsbullet
$\TIMEabs: \mathcal{P}_\mathbb{D} \rightarrow \mathcal{P}_\mathbb{T}, \TIMEabs(s) := z$
is equivalent to its concrete counterpart; only the data domain is extended:
$z(t) = t$ if $t \in \Ticks(s)$ and $z(t) = s(t)$ otherwise.

\KWabsbullet
$\LIFTabs: ({\mathbb{D}_1}^\#_\bot \times \dots
\times {\mathbb{D}_n}^\#_\bot \rightarrow \mathbb{D}^\#_\bot)
\rightarrow (\mathcal{P}_{\mathbb{D}_1} \times \dots \times
\mathcal{P}_{\mathbb{D}_n} \rightarrow \mathcal{P}_\mathbb{D}),
\LIFTabs(f^\#)(s_1, \ldots, s_n) := z$ can be defined similarly to
its concrete counterpart, because the abstract function $f^\#$ takes
care of the gaps:
\[
z(t) =\begin{cases}
  f^\#(s_1(t), \dots, s_n(t)) & \text{if $s_1\neq\mathrm{?},\dots, s_n\neq\mathrm{?}$} \\
  ? & \text{otherwise}
\end{cases}
\]
The operator $\LIFTabs$ is restricted to those functions $f^\#$ that
are an abstraction of functions $f$ that can be used in \LIFT, that is,
$f(\bot,\dots,\bot) = \bot$.
Using the abstract lift we can derive the abstract counterparts of
$\DEF{const}$ and $\DEF{merge}$:

\KWlightbullet
$\DEF{const}^\#(c)(a) := \LIFTabs(f_c)(a)$ with $f_c(d) := c$ if
$d \neq \gap$ and $f_c(\gap) := \gap$ otherwise
maps all events' values to a constant while preserving the gaps.
Using $\DEF{const}^\#$ we can define constant signals without any gaps, e.g.,
$\DEF{true}^\# := \DEF{const}^\#(\DEF{true})(\KW{unit}^\#)$ or
$\DEF{zero}^\# := \DEF{const}^\#(0)(\KW{unit}^\#)$.

\KWlightbullet
$\DEF{merge}^\#(x, y) := \LIFTabs(f)(x, y)$ with
$f(a \not\in \{\gap,\bot\}, b) = a$,
$f(⊥, b) = b$,
$f(\gap, b \in \{\gap,\bot\}) = \gap$, and
$f(\gap, b \not\in \{\gap,\bot\}) = \top$.

\begin{wrapfigure}[3]{r}{40mm}
\vskip-8mm\raggedleft
\begin{streampicture}[scale=.7]
  \begin{stream}{x}
    \definedto{4}
    \eventat[red]{2}
    \eventat[red]{3}
    \undefinedto{6}
    \definedto{10}
    \eventat[red]{8}
    \gapat{9}
  \end{stream}
  \begin{stream}{y}
    \definedto{5}
    \eventat[blue]{1}
    \eventat[blue]{3}
    \undefinedto{7}
    \definedto{10}
    \gapat{8}
    \eventat[blue]{9}
  \end{stream}
  \begin{stream}{z}
    \definedto{4}
    \eventat[blue]{1}
    \eventat[red]{2}
    \eventat[red]{3}
    \undefinedto{7}
    \eventat[red]{8}
    \topat{9}
    \definedto{10}
  \end{stream}
\end{streampicture}
\end{wrapfigure}

The diagram on the right shows an example trace merging the events of
the streams $x$ and $y$. 
The symbol $\circ$ indicates a point-wise gap.
Note how an event on the first stream takes precedence over a gap on
the second stream, but not the other way round, similarly to how
events from the first stream are prioritized if both streams have an
event at the same timestamp.

\KWabsbullet
$\LASTabs: \mathcal{P}_{\mathbb{D}_1} \times \mathcal{P}_{\mathbb{D}_2}
\rightarrow \mathcal{P}_{\mathbb{D}_1}, \LASTabs(v,r) := z$
has three major extensions over its concrete counterpart:
\begin{compactenum}[(1)]
\item $\top$ is added as an output in case an event on $r$
  occurs and there were events on the stream $v$ of values
  but all followed by a gap. 
\item $\gap$ is outputted for all gaps on the stream $r$ of trigger events if
  there have been events on the stream $v$ of values.
\item $\gap$ can also be output if an event occurs on $r$
  and no event occurred on $v$ before except for a
  gap.
\end{compactenum}
The parts similar to the concrete operator are typeset in gray:
\[
z(t) = \left\{
\begin{array}{l@{\ \ }l@{\ \ }l}
  d & \color{gray} t\in\Ticks(r) \wedge \exists_{t' < t}\AUX{isLast}(t,t',v,d) \color{black} \wedge  \forall_{t''|t'<t''<t}v(t'')\neq\gap\\
  \top & t\in\Ticks(r) \wedge \exists_{t' < t}\AUX{isLast}(t,t',v,\gap) \wedge  \exists_{t''|t'<t''<t}v(t'')=\gap & (1)\\
  \bot & \color{gray} r(t) = \bot \wedge  \AUX{defined}(z,t) \vee \forall_{t'<t} v(t') = \bot \\
  \gap & \AUX{defined}(z,t) \wedge r(t)=\gap \wedge \exists_{t'<t}v(t')\neq\bot & (2) \\
  \gap & \AUX{defined}(z,t) \wedge t\in\Ticks(r) \wedge \forall_{t'<t}t'\notin\Ticks(v) \wedge \exists_{t'<t}v(t')=\gap & (3) \\
   \textrm{?} & \color{gray}\text{otherwise}
\end{array}
\right.
\]

\begin{wrapfigure}[4]{r}{40mm}
\vskip-8mm\raggedleft
\begin{streampicture}[y=7mm,x=7mm]
  \begin{stream}{v}
    \undefinedto{2}
    \definedto{7}
    \eventat{3}
    \gapat{4}
    \eventat{6}
    \undefinedto{8}
    \definedto{10}
  \end{stream}
  \begin{stream}{r}
    \definedto{8}
    \eventat{1}
    \eventat{5}
    \eventat{6}
    \eventat{7}
    \undefinedto{9}
    \definedto{10}
  \end{stream}
  \begin{stream}{z}
    \definedto{8}
    \gapat{1}
    \topat{5}
    \topat{6}
    \eventat{7}
    \undefinedto{9}
    \definedto{10}
  \end{stream}
  \node[every label] at (1,-8) {(3)};
  \node[every label] at (5.5,-8) {(1)};
  \node[every label] at (7,-8) {(d)};
  \node[every label] at (8.5,-8) {(2)};
\end{streampicture}
\end{wrapfigure}

The trace diagram on the right shows an example trace covering most
edge cases of the abstract last. 
The output stream $z$ is a point-wise gap if triggered after
initial gaps (3); $z$ is $\top$ if triggered after non-initial gaps (1);
$z$ is an event if triggered after a gapless sequence (d); and $z$
inherits all gaps from the stream of trigger events (2).

We can now combine the $\LASTabs$ and the $\LIFTabs$ operators to
realize:

\KWlightbullet
\emph{abstract signal lift} for total functions $f:
\mathbb D \times \mathbb D' \to \mathbb D''$ as $\DEF{slift}^\#(f)(x,
y) := \LIFTabs(g_f)(x', y')$ with $x' := \DEF{merge}^\#(x, \LAST^\#(x,
y))$ and $y' := \DEF{merge}^\#(y, \LAST^\#(y, x))$, as well as $g_f(a
\not\in \{\gap,\bot\}, b \not\in \{\gap,\bot\}) = f(a, b)$, $g_f(⊥, b)
= g_f(a, ⊥) = ⊥$, $g_f(\gap, \gap) = \gap$, and $g_f(\gap, b
\not\in\{\gap,\bot\}) = g_f(a \not\in\{\gap,\bot\}, \gap) = \gap$.

\begin{wrapfigure}[6]{r}{6cm}
\vskip-9mm\raggedleft
  \begin{streampicture}[yscale=1.4,xscale=.68]
    \begin{stream}{\DEF{values}}
      \definedto{8}
      \eventat[val=3]{1}
      \eventat[val=2]{3}
      \eventat[val=4]{4}
      \gapat{5}
      \eventat[val=5]{7}
      \undefinedto{10}
      \definedto{14}
      \eventat[val=3]{11}
      \eventat[val=6]{12}
      \eventat[val=1]{13}
    \end{stream}
    \streamskip[-1]
    \begin{stream}{\DEF{resets}}
      \definedto{14}
      \eventat{2}
      \eventat{6}
      \eventat{9}
      \eventat{12}
    \end{stream}
    \begin{stream}{\DEF{cond}}
      \definedto{8}
      \eventat[val=tt]{2}
      \eventat[val=ff]{3}
      \eventat[val=ff]{4}
      \gapat{5}
      \topat{6}
      \eventat[val=ff]{7}
      \undefinedto{10}
      \topat{9}
      \definedto{14}
      \eventat[val=ff]{11}
      \eventat[val=tt]{12}
      \eventat[val=ff]{13}
    \end{stream}
    \begin{stream}{\DEF{sum}}
      \definedto{8}
      \eventat[val=0]{2}
      \eventat[val=2]{3}
      \eventat[val=6]{4}
      \gapat{5}
      \topat{6}
      \topat{7}
      \undefinedto{10}
      \topat{9}
      \definedto{14}
      \topat{11}
      \eventat[val=0]{12}
      \eventat[val=1]{13}
    \end{stream}
  \end{streampicture}
\end{wrapfigure}

\refstepcounter{example}
\paragraph{Example \theexample.}
\label{abstract-running-example}
By replacing every TeSSLa operator in \autoref{running-example} with
their abstract counterparts and applying it to the abstract input streams
$\DEF{values} \in 𝓟_{\mathbb Z}$ and $\DEF{resets} \in 𝓟_{\mathbb U}$,
we derive the abstract stream $\DEF{cond} \in 𝓟_{\mathbb B}$ and the recursively
derived abstract stream $\DEF{sum} \in 𝓟_{\mathbb Z}$:
After the large gap on $\DEF{values}$, the $\DEF{sum}$ stream
eventually recovers completely.
The first reset after the point-wise gap does not lead to full recovery,
because at that point the last event on \DEF{values} cannot be accessed,
because of the prior gap.
The next reset falls into the gap, so again $\DEF{cond}$ cannot be
evaluated.
In a similar fashion one can define an abstract \KWabs{delay} operator as
counterpart of the concrete \KW{delay}.
See Appendix~\ref{app:delay} for details.

Following from the definitions of the abstract \tessla operators we
get:

\begin{theorem}
  \label{thm:abstract-tessla}
  Every abstract \tessla operator is an abstraction of its concrete
  counterpart.
\end{theorem}

Theorem~\ref{thm:abstract-tessla} implies that abstract \tessla
operators are sound in the following way.
Let $o$ be a concrete \tessla operator with the abstract counterpart
$o^\#$ and let $s \in \mathcal P_{\mathbb D}$ be an abstract event
stream with a concretization function $\gamma$.  
Then, \( o(\gamma(s)) \preceq \gamma(o^\#(s)). \)
Since abstract interpretation is compositional we can directly follow
from the above theorem:
\begin{corollary}
  \label{cor:abstract-tessla}
  If a concrete TeSSLa specification $\phi$ is transformed into a
  specification $\psi$ by replacing every concrete operator in $\phi$ with its
  abstract counterpart, then $\psi$ is an abstraction of $\phi$.
\end{corollary}

Theorem~\ref{thm:abstract-tessla} guarantees that applying 
abstract \tessla operators to the abstract event stream is still sound
regarding the underlying set of possible concrete event streams.
However, we have established no result so far about the accuracy of
the abstract \tessla operators.
The abstraction returning only the completely unknown stream ($\Delta =
\emptyset$) is sound but useless.
The following theorem states, that our abstract \tessla operators are
optimal in terms of accuracy.
Using a perfect abstraction guarantees the abstract \tessla operators
preserve as much information as can possibly be encoded in the
resulting abstract event streams.

\begin{theorem}
  \label{theorem:abstraction}
  Every abstract \tessla operator is a perfect abstraction of its
  concrete counterpart.
\end{theorem}

Given a concrete \tessla operator $o$ and its abstract counterpart
$o^\#$, and any abstract event stream $s \in \mathcal P_{\mathbb D}$
with the Galois Connection $(\alpha,\gamma)$ between $2^{\mathcal
  S_{\mathbb D}}$ and $\mathcal P_{\mathbb D}$ one can show that
\( o^\#(s) = \alpha(o(\gamma(s)). \)
Applying the abstract operator on the abstract event stream is as good
as applying the concrete operator on every possible event stream
represented by the abstract event stream. 
Thus $o^\#$ is a perfect abstraction of $o$.
(The detailed proof can be found in Appendix \ref{app:proof}.)
Note that we assume that $f^\#$ is a perfect abstraction of $f$ to
conclude that $\LIFTabs(f^\#)$ is a perfect abstraction of $\LIFT(f)$.

In Corollary~\ref{cor:abstract-tessla} we have shown that a specification $\psi$
(generated by replacing the concrete TeSSLa operator in $\phi$ with their abstract
counterparts) is an abstraction of $\phi$. Note that $\psi$ is in general not a
perfect abstraction of $\phi$. We study some special cases of perfect abstractions of
compositional specifications in Section~\ref{sec:compositional}.

The next result states that the abstract operators can be defined in
terms of concrete \tessla operators. Realizing the abstract
operators in TeSSLa does not require an enhancement in the expressivity of
\tessla.

\begin{theorem}
  \label{thm:realizability}
  The semantics of the abstract \tessla operators can be encoded in
  \tessla using only the concrete operators.
\end{theorem}

\begin{proof}
  One can observe that the abstract TeSSLa operators are monotone
  and future independent (the output stream up to $t$ only depends on
  the input streams up to $t$.)
  As shown in~\cite{tessla}, \tessla can express every such function. \qed
\end{proof}

\subsection{Fixpoint Calculations Ensuring Well-Formedness}
\label{sec:well-formedness}

A concrete TeSSLa specification consists of stream variables and
possibly recursive equations applying concrete \tessla operators to
the stream variables.
Theorem~\ref{thm:abstract-tessla} and Corollary~\ref{cor:abstract-tessla}
guarantee that a concrete
\tessla specification can be transformed into an abstract \tessla
specification, which is able to handle gaps in the input streams.
Additionally, Theorem~\ref{thm:realizability} states that the
abstract \tessla operators can be implemented using concrete
\tessla operators.
Combining these two results, one can transform a given concrete
specification $\phi$ into a corresponding specification $\psi$, which
realizes the abstract \tessla semantics of the operators in $\phi$,
but only uses concrete \tessla operators.

However, using the realization of the abstract \tessla operators in TeSSLa adds
additional cyclic dependencies in $\psi$ between the stream variables.
A TeSSLa specification is well-formed if every cycle of its dependency
graph contains at least one edge guarded by a last (or a delay)
operator, which is required to guarantee the existence of a unique
fixed-point and hence computability (see \cite{tessla}).

\begin{wrapfigure}[3]{r}{35mm}
\vskip-8mm\raggedleft
\begin{streampicture}[yscale=.6]
  \begin{stream}{v}
    \definedto{4}
    \eventat{2}
  \end{stream}
  \begin{stream}{r}
    \definedto{1}
    \undefinedto{3}
    \definedto{4}
  \end{stream}
  \begin{stream}{\LASTabs(v,r)}
    \definedto{2}
    \undefinedto{3}
    \definedto{4}
  \end{stream}
  \draw[red, semithick, line cap=round, rounded corners=1.5pt,-{Latex[round, length=4pt]}, shorten >=1.5pt, shorten <=1.5pt] (2,-6) -- ++(0,-.8) -- ++(.5,0) -- ++(0,6.3) -- ++(-.5,0) -- ++(0,-1.5);
\end{streampicture}
\end{wrapfigure}

Consider the trace diagram on the right showing
$\LASTabs(v,r)$.  
If $v$ is used in a recursive manner, i.e., $v$ is defined in terms of
$\LASTabs(v,r)$, then the first event on $v$ could start a gap on
$\LASTabs(v,r)$ that could start a gap on $v$ at the same timestamp.
As a result $v$ has an unguarded cyclic dependency and hence the
specification is not well-formed.  
To overcome this issue one can split up the value and gap calculation
sequentially, reintroducing guards in the cyclic dependency:
\begin{definition}[Unrolled Abstract Last] \label{def:unrolling}
  We define two variants of the abstract last, $\LASTbot$ and
  $\LASTgap$ as follows.
  Let $z = \LASTabs(v,r)$, then $\LASTbot(v,r) := z_\bot$ and
  $\LASTgap(v,r,d) := z_{\gap}$.
  \[z_\bot(t) = 
  \begin{cases}
    z(t) & \text{if } z(t) \neq \gap \\
    \bot & \text{otherwise}
  \end{cases}
  \qquad
  z_{\gap}(t) = \begin{cases}
    d(t) & \text{if } t \in \Ticks(d) \\
    \gap & \text{if } t \notin \Ticks(d) \land z(t) = \gap \\
    \bot & \text{otherwise}
  \end{cases}\]
\end{definition}
Function \LASTbot executes a normal calculation of the events, in
the same way an abstract last would do, but neglecting gaps and
outputting $\bot$ as long as there is no event. 
Function \LASTgap takes a third input stream and outputs its events
directly, but calculates gaps correctly as $\LASTabs$ would do.

Since the trigger input of a \LAST operator cannot be recursive in a
well-formed specification, a recursive equation using one last has the
form $x = \LASTabs(v,r)$ and $v = f(x,\vec c)$, where $\vec c$ is a
vector of streams not involved in the recursion and $f$ does not
introduce further last (or delay) operators.
Now, this equation system can be rewritten in the following equivalent form:
\[
  x' = \LASTbot(v,r)
  \qquad
  v' = f(x',\vec c)
  \qquad
  x = \LASTgap(v',r,x')
  \qquad
  v = f(x,\vec c)
\]
This pattern can be repeated if multiple recursive abstract lasts
are used and can also be applied in a similar fashion to mutually
recursive equations and the delay operator.

%% file: compositional.tex

\section{Perfection of Compositional Specifications}
\label{sec:compositional}

A concrete \tessla specification $\phi$ can be transformed into an
abstract \tessla specification $\psi$ by replacing the concrete
operators with their abstract counterparts.
For two functions $f$ and $g$ with corresponding abstractions $f^\#$ and $g^\#$
the function composition $f^\#\circ g^\#$ is an abstraction of $f\circ g$.
Unfortunately, even if $f^\#$ and $g^\#$ are perfect abstractions,
$f^\#\circ g^\#$ is not necessarily a perfect abstraction.
Hence, $\psi$ needs not be a perfect abstraction of $\phi$.
In this section we discuss the perfection of two common compositional
\tessla operators:
(1) the $\DEF{slift}^\#$ defined in Section~\ref{sec:abstraction} is a
composition of \LASTabs in $\LIFTabs$\hspace{-0.25em}, which realizes signal
semantics;
(2) $\LASTabs(\TIMEabs(v), r)$, which is a common pattern used when
comparing timestamps.

The $\DEF{slift}^\#$ is defined as the $\LIFTabs$ applied to
the synchronized versions $x'$ and $y'$ of the input streams $x$ and $y$.
The input stream $x$ is synchronized with $y$ by keeping the
original events of $x$ and reproducing the last known value of
$x$ for every timestamp with an event on $y$, but not on $x$.

\begin{theorem} 
  If $f^\#$ is a perfect abstraction of $f$ then
  $\DEF{slift}(f^\#)^\#$ is a perfect abstraction of $\DEF{slift}(f)$.
\end{theorem}

\begin{proof}
  Since $\DEF{slift}^\#$ is defined on abstract event streams we need
  to consider gaps.
  The stream $x'$ does not have any gap or event until the first gap
  or event on $x$. 
  After the first gap or event on $x$ the synchronized stream $x'$
  contains a gap or event at every timestamp where $x$ or $y$ contain
  a gap or event. 
  Because $\DEF{slift}^\#$ is symmetric in terms of the event pattern
  the same holds for $y'$.
  By definition, $\DEF{slift}^\#(f^\#)(x,y) = z$ contains an event or
  gap iff $x'$ and $y'$ contain an event or gap, because $f$ is
  a total function.
  The output stream $z$ contains an event iff $x'$ and $y'$ contain
  events. 
  The events values are ensured to be as precise as possible, because
  $f^\#$ is a perfect abstraction of $f$.
  \qed
\end{proof}

\begin{wrapfigure}[6]{r}{57mm}
\vskip-8mm\raggedleft
\begin{streampicture}
  \draw[very thin, draw=gray] (1,-2) -- (1,-8.5) node[below, every label, red, inner sep=0pt] {\strut$a$};
  \draw[very thin, draw=gray] (4,-2) -- (4,-8.5) node[below, every label, red, inner sep=0pt] {\strut$b$};
  \begin{stream}{v}
    \definedto{2}
    \eventat{1}
    \undefinedto{4}
    \definedto{6}
  \end{stream}
  \begin{stream}{r}
    \definedto{6}
    \eventat{5}
  \end{stream}
  \begin{stream}{\LASTabs(\TIMEabs(v), r)}
    \definedto{6}
    \topat{5}
  \end{stream}
  \begin{stream}{\LASTTIMEabs(v, r)}
    \definedto{6}
    \eventat{5}
    \node[every label] at (5,-7.2) {$[\color{red}a\color{black},\color{red}b\color{black}]$};
  \end{stream}
\end{streampicture}
\end{wrapfigure}

TeSSLa allows arbitrary computations on the timestamps of events using
the \TIME operator.  
The specification $z = \TIME(v)$ derives a stream $z$ from $v$ by
replacing all event's values in $v$ with the event's timestamps.
The stream variable $z$ can now be used in any computation expressible
in \tessla. 
Hence, \tessla does not distinguish between timestamps and other
values, and consequently abstract \tessla specifications cannot make
use of the monotonicity of time.
As an example consider the trace diagram on the right. 
The stream $\LASTabs(\TIMEabs(v), r)$ is derived from $v$ by composing
\TIMEabs and \LASTabs. 
Since \TIMEabs changes the events values with their timestamps, the
\LASTabs does not know any longer that we are interested in the last
timestamp of $v$ and can only produce an event with the value $\top$
representing all possible values.
To overcome this issue we define $\LASTTIME(v, r) := \LAST(\TIME(v),
r)$ and provide a direct abstraction, which allows a special treatment
of timestamps.
\begin{definition}[Time Aware Abstract Last]
  Let $y = \LASTabs(\TIMEabs(v),r)$, we define $\LASTTIMEabs: \mathcal
  P_{\mathbb D} \times \mathcal P_{\mathbb D'} \to \mathcal
  P_{2^{\mathbb T}}, \LASTTIMEabs(v, r) := z$ as
  $z(t) = [a, b]$ if $y(t) = \top$ with
  $a = \op{inf}\{t' < t \mid \forall_{t' < t'' < t} v(t'') \neq \gap\}$ and $b = \op{max}\{t' < t \mid t' \in \Ticks(v)\}$ and $z(t) = y(t)$ otherwise.
\end{definition}
Now the following result holds
(the proof can be found in Appendix \ref{app:proof}).
\begin{theorem}
  \label{theorem:timelast}
  $\LASTTIMEabs$ is a perfect abstraction of $\LASTTIME$.
\end{theorem}
A similar problem occurs if $\DEF{slift}^\#$ is used to compare
event's timestamps.
In Example~\ref{abstract-running-example} the stream $\DEF{cond}$
derived by comparing the timestamps of $\DEF{values}$ and $\DEF{resets}$
has two events with the unknown data value $\top$ because of prior
gaps on $\DEF{values}$.
Since the $\DEF{slift}^\#$ is defined in terms of $\LIFTabs$ and $\LASTabs$
we can define the function $\DEF{sliftTime}^\#(f^\#)(x,y)$ as an abstraction
for the special case
$\DEF{sliftTime}(f)(x,y) = \DEF{slift}(f)(\TIME(x), \TIME(y))$
by using
$\DEF{lastTime}^{\#}$ instead of $\DEF{last}^\#$ and ensuring that $f^\#$
uses interval arithmetics to abstract $f$.
Note that $\DEF{sliftTime}^\#(f^\#)$ is a perfect abstraction of
$\DEF{sliftTime}(f)$.

\begin{wrapfigure}[6]{r}{58mm}
\vskip-8mm\raggedleft
  \begin{streampicture}[yscale=1.4,xscale=.68]
    \begin{stream}{\DEF{values}}
      \definedto{8}
      \eventat[val=3]{1}
      \eventat[val=2]{3}
      \eventat[val=4]{4}
      \gapat{5}
      \eventat[val=5]{7}
      \undefinedto{10}
      \definedto{14}
      \eventat[val=3]{11}
      \eventat[val=6]{12}
      \eventat[val=1]{13}
    \end{stream}
    \streamskip[-1]
    \begin{stream}{\DEF{resets}}
      \definedto{14}
      \eventat{2}
      \eventat{6}
      \eventat{9}
      \eventat{12}
    \end{stream}
    \begin{stream}{\DEF{cond}}
      \definedto{8}
      \eventat[val=tt]{2}
      \eventat[val=ff]{3}
      \eventat[val=ff]{4}
      \gapat{5}
      \eventat[red,val=tt]{6}
      \eventat[val=ff]{7}
      \undefinedto{10}
      \eventat[red,val=tt]{9}
      \definedto{14}
      \eventat[val=ff]{11}
      \eventat[val=tt]{12}
      \eventat[val=ff]{13}
    \end{stream}
    \begin{stream}{\DEF{sum}}
      \definedto{8}
      \eventat[val=0]{2}
      \eventat[val=2]{3}
      \eventat[val=6]{4}
      \gapat{5}
      \eventat[red,val=0]{6}
      \eventat[red,val=5]{7}
      \undefinedto{10}
      \eventat[red,val=0]{9}
      \definedto{14}
      \topat{11}
      \eventat[val=0]{12}
      \eventat[val=1]{13}
    \end{stream}
  \end{streampicture}
\end{wrapfigure}

\refstepcounter{example}
\paragraph{Example \theexample.}
\label{example:time-abstract-tessla}
To illustrate the perfect abstraction
$\DEF{sliftTime}^\#$ we update the definition of \DEF{cond} in
Example~\ref{abstract-running-example} as follows: $\DEF{cond} =
\DEF{sliftTime}(\le)(\DEF{resets}, \DEF{values})$.  The events drawn
in red now have concrete values instead of $\top$ as in
Example~\ref{abstract-running-example}.


%% file: queue.tex

\section{Abstractions for Sliding Windows}
\label{sec:queue}

In this section we demonstrate how to apply the techniques
presented in this paper to specifications
with richer data domains.
In particular, we show now a TeSSLa specification that uses a queue to
compute the average load of a processor in the last five time units.
The moving window is realized using a queue storing all events that happened in the time window.
The stream $\DEF{load}\in \mathcal{S}_\mathbb{R}$ contains an event
every time the input load changes:
\begin{align*}
  \DEF{stripped} &= \DEF{slift}(\AUX{remOlder}_5)(\TIME(\DEF{load}), \DEF{merge}(\LAST(\DEF{queue}, \DEF{load}), \langle\rangle)))\\
  \DEF{queue} &= \LIFT(\AUX{enq})(\TIME(\DEF{load}), \DEF{load}, \DEF{stripped})\\
  \DEF{avg} &= \LIFT(\AUX{int})(\DEF{queue}, \TIME(\DEF{load}))\\
  \AUX{int}(q, u) &= \AUX{fold}(f, q, 0, u) \qquad
  f(a, b, v, \AUX{acc}) = \AUX{acc} + v \cdot (b-a)/5
\end{align*}
The queue operation \AUX{enq} adds elements to the queue, while
$\AUX{remOlder}_5$ removes elements with a timestamp older than five
time units. 
The function \AUX{int} accumulates all values in the queue weighted
by the length of the corresponding signal piece.
The queue operation \AUX{fold} is used to fold the function $f$
over all elements from the queue with the initial accumulator 0
until the timestamp $u$.
Hence $f$ is called for every element in the queue with the timestamps $a$
and $b$, the element's value $v$ and the accumulator.
Consequently, the specification adds elements to the queue, removes
the expired elements and accumulates the remaining values.
Using our approach we replace every operator with its abstract
counterpart and represent abstract queues appropriately such that also
queues with partly unknown entries can be modeled.  By doing this we
obtain a specification that is able to handle gaps in the input
stream, as illustrated in Fig.~\ref{fig:queue}.

\begin{figure}[t]
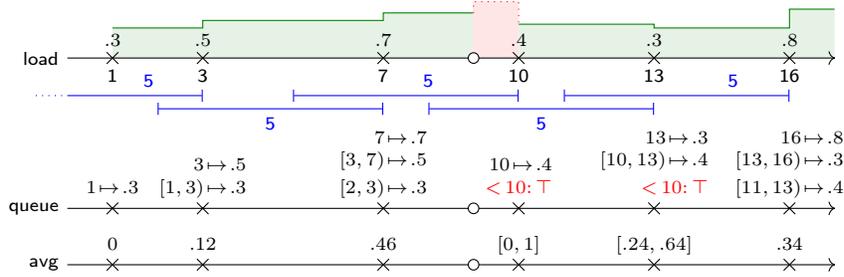

  \let\mathmapsto\mapsto
  \renewcommand{\mapsto}{\text{\hspace{.1em}$\mathmapsto$\hspace{.1em}}}
  \centering
  \begin{streampicture}[xscale=1.2]
    \begin{stream}[time stamps]{\DEF{load}}
      \fill[yshift=-\thestreamycoordinate cm, green!50!black!10]
        (1,0) -- (1,1.6) -- (3,1.6) -- (3,2) -- (7,2) -- (7,2.4) -- (9,2.4) -- (9,0);
      \draw[yshift=-\thestreamycoordinate cm, green!50!black]
        (1,1.6) -- (3,1.6) -- (3,2) -- (7,2) -- (7,2.4) -- (9,2.4);
      \fill[yshift=-\thestreamycoordinate cm, red!10]
        (9,0) -- (9,3) -- (10,3) -- (10,0) -- cycle;
      \draw[yshift=-\thestreamycoordinate cm, red, dotted]
        (9,2.4) -- (9,3) -- (10,3) -- (10,1.8);
      \fill[yshift=-\thestreamycoordinate cm, green!50!black!10]
        (10,0) -- (10,1.8) -- (13,1.8) -- (13,1.6) -- (16,1.6) -- (16,2.6) -- (17,2.6) -- (17,0);
      \draw[yshift=-\thestreamycoordinate cm, green!50!black]
        (10,1.8) -- (13,1.8) -- (13,1.6) -- (16,1.6) -- (16,2.6) -- (17,2.6);
      \definedto{17}
      \eventat[val=.3]{1}
      \eventat[val=.5]{3}
      \eventat[val=.7]{7}
      \gapat{9}
      \eventat[val=.4]{10}
      \eventat[val=.3]{13}
      \eventat[val=.8]{16}
    \end{stream}
      \path[blue,|-|,yshift=-\thestreamycoordinate cm,font=\scriptsize\sffamily]
        (-.7,-2) edge[dotted,-] (0,-2)
        (0,-2) edge[-|] node[above,pos=.6] {5} (3,-2)
        (2,-2.7) edge node[below] {5} (7,-2.7)
        (5,-2) edge node[above,pos=.6] {5} (10,-2)
        (8,-2.7) edge node[below] {5} (13,-2.7)
        (11,-2) edge node[above,near end] {5} (16,-2);
    \streamskip[6]
    \newcommand{\queueeventat}[2]{
      \eventat[val={\shortstack[r]{#2}}]{#1}
    }
    \begin{stream}[every label/.append style={inner sep=1pt}]{\DEF{queue}}
      \definedto{17}
      \queueeventat{1}{$\strut1\mapsto.3$}
      \queueeventat{3}{$3\mapsto.5$\\\strut$[1,3)\mapsto.3$}
      \queueeventat{7}{$7\mapsto.7$\\$[3,7)\mapsto.5$\\\strut$[2,3)\mapsto.3$}
      \gapat{9}
      \queueeventat{10}{$10\mapsto.4$\\\strut\color{red}$<$\,10:\,$\top$}
      \queueeventat{13}{$13\mapsto.3$\\$[10,13)\mapsto.4$\\\strut\color{red}$<$\,10:\,$\top$}
      \queueeventat{16}{$16\mapsto.8$\\$[13,16)\mapsto.3$\\\strut$[11,13)\mapsto.4$}
    \end{stream}
    \streamskip
    \begin{stream}[every label/.append style={inner sep=1pt}]{\DEF{avg}}
      \definedto{17}
      \eventat[val=\strut0]{1}
      \eventat[val=\strut.12]{3}
      \eventat[val=\strut.46]{7}
      \gapat{9}
      \eventat[val=\strut{$[0,1]$}]{10}
      \eventat[val=\strut{$[.24,.64]$}]{13}
      \eventat[val=\strut.34]{16}
    \end{stream}
  \end{streampicture}
  \caption{Example trace of the abstract queue specification.}
  \label{fig:queue}
\end{figure}

We can extend the example such that the queue only holds a
predefined maximum number of events (to guarantee a finite state
implementation).
When removing events we represent these as unknown entries in the abstract queues.
The abstract $\AUX{fold}^\#$ is capable of
computing the interval of possible average loads
for queues with unknown elements anyhow.

Note that the average load is only updated for every event on the input stream.
Using a delay operator, we can set a
timeout whenever an element leaves the sliding window in the abstract
setting.
The element is removed from the queue at that timeout and the new value
of the queue is updated with the remaining elements.
%
%
Formal definitions of the queue functions as well as the complete
specifications are available
online\footnote{\url{http://tessla-a.isp.uni-luebeck.de/}}.

%% file: empirical.tex

\section{Implementation and Empirical Evaluation}
\label{sec:evaluation}

As discussed in Section~\ref{sec:well-formedness} the abstract TeSSLa
operators can be implemented using only the existing concrete TeSSLa
operators.
We implemented the abstract TeSSLa operators as macros specified in
the TeSSLa language itself such that the existing TeSSLa engine
presented in \cite{tessla} can handle abstract TeSSLa specifications.
An abstract event stream $(s, \Delta) \in \mathcal P_{\mathbb D}$ can
be represented as two TeSSLa streams $s \in \mathcal S_{\mathbb D^\#}$
and $s_d \in \mathcal S_{X}$, where $X$ contains the following six
possible changes of $\Delta$: inclusive start, exclusive start,
inclusive end, exclusive end, point-wise gap and point-wise event in a
gap.
Using this encoding it is sufficient to look up the latest $s_d(t')$
with $t' \le t$ to decide whether $t \in \Delta$.
While this encoding already allows a decent implementation of abstract
TeSSLa we go one step further and assume a finite time domain with a
limited precision, e.g., 64 bit integers or floats.
Under this assumption there is always a known smallest relative
timestamp $\epsilon$.
Hence, we can use the encoding $s_d \in \mathcal S_{\mathbb B}$ where
an event $s_d(t) = \mathrm{true}$ encodes a start inclusive and
$s_d(t) = \mathrm{false}$ an end exclusive.
This encoding captures the most common cases and simplifies the
implementation of union and intersection on $\Delta$ enormously since
they can now be realized as $\DEF{slift}(\vee)$ and
$\DEF{slift}(\wedge)$, resp.
The other possible switches at timestamp $t$ can be represented as
follows: $s_d(t+\epsilon) = \mathrm{true}$ encodes an exclusive start,
$s_d(t+\epsilon) = \mathrm{false}$ encodes an inclusive end,
$s_d(t) = \mathrm{true}$ and $s_d(t+\epsilon) = \mathrm{false}$
encodes a point-wise event in a gap, and $s_d(t) = \mathrm{false}$ and
$s_d(t+\epsilon) = \mathrm{true}$ encodes a point-wise gap.
Using this encoding the abstract TeSSLa operators do not need to
handle these additional cases explicitly.

Furthermore, assuming the smallest relative timestamp $\epsilon$, we
can avoid the need to perform the unrolling defined in
Defs.~\ref{def:unrolling} by delaying the
second part of the computation to the next possible timestamp
$t+\epsilon$.

As a final efficiency improvement we simplified $\LASTabs$ before the
first event on the stream of values, which are not relevant in
practice.
The abstract operator and hence abstract specifications are of course
still a sound abstraction of their concrete counterparts, but due to
over-abstractions no longer a perfect one during this initial
event-less phase of the stream of values.

The implementation in form of a macro library for the existing TeSSLa
engine is available together with all the examples and scripts used in
the following empirical evaluation and can be experimented with in a
web IDE\footnote{\url{http://tessla-a.isp.uni-luebeck.de/}}.

In the following empirical evaluation we measure the accuracy of the
abstractions presented in this paper.
An abstract event stream represents input data with some sequences of
data loss, where we do not know if any events might have been occurred
or what their values have been.
Applying an abstract TeSSLa specification to such an input stream
takes these gaps into account and provides output streams that in turn
contain sequences of gaps and sequences containing concrete events.
To evaluate the accuracy of this procedure we compare the output of an
abstract TeSSLa specification with the best possible output.

\begin{wrapfigure}[7]{r}{40mm}
\vskip-8mm\raggedleft
\begin{tikzpicture}[
    on grid, auto,
    label/.style={
      inner sep=1pt
    },
    abstract/.style={
      draw=blue,
      fill=blue!10,
      draw,
      minimum height=6mm,
      minimum width=8mm,
      rounded corners=2pt
    },
    concrete/.style={
      draw=red,
      fill=red!10,
      draw,
      minimum height=6mm,
      minimum width=8mm,
      rounded corners=2pt
    }
  ]
  \node[abstract] (abstract input) {$\mathcal P_\mathbb D$};
  \node[abstract, right=1.5cm of abstract input] (abstract output) {$\mathcal P_\mathbb D$};
  \node[concrete, below=1.5cm of abstract input] (concrete input) {$2^{\mathcal S_\mathbb D}$};
  \node[abstract, below=1.2cm of abstract output] (concrete output) {$2^{\mathcal S_\mathbb D}$};
  \node[concrete, below=7mm of concrete output] (perfect concrete output) {$2^{\mathcal S_\mathbb D}$};
  \node[abstract, right=1.5cm of concrete output] (ignorance) {$\mathbb I$};
  \node[concrete, below=7mm of ignorance] (perfect ignorance) {$\mathbb I$};
  \path[->]
    (abstract input) edge node[label] {$\phi^\#$} (abstract output)
    (abstract input) edge node[label] {$\gamma$} (concrete input)
    (abstract output) edge node[label] {$\gamma$} (concrete output)
    (concrete input) edge node[label, near start] {$\phi$} (perfect concrete output)
    (perfect concrete output) edge node[label] {$\iota$} (perfect ignorance)
    (concrete output) edge node[label] {$\iota$} (ignorance);
\end{tikzpicture}
\end{wrapfigure}

Let $r \in \mathcal P_{\mathbb D}$ be an abstract event stream.
We obtain the set $R$ of all possible input streams containing all
possible variants that might have happened during gaps in $r$ by
applying the concretization function $\gamma$ on the abstract input
stream.
Now we can apply the concrete TeSSLa specification $\phi$ to all
streams in $R$ and get the set $S$ of concrete output streams.
On the other hand we apply the abstract TeSSLa specification $\phi^\#$
directly to $r$ and get the abstract output stream $s$.
Now $S$ is the set of all possible output streams and $\gamma(s)$ is
the set of output streams defined by the abstract TeSSLa
specification.
The diagram on the right depicts this comparison process.

\begin{wrapfigure}[6]{r}{40mm}
\vskip-8mm\raggedleft
\begin{streampicture}[y=12mm]
  \begin{stream}{a}
    \definedto{6}
    \eventat[val=0]{0}
    \eventat[val=2]{1}
    \eventat[val=1]{2}
  \end{stream}
  \begin{stream}{b}
    \definedto{6}
    \eventat[val=0]{0}
    \eventat[val=2]{1}
    \eventat[val=0]{3}
    \eventat[val=1]{4}
  \end{stream}
  \begin{stream}{c}
    \definedto{6}
    \eventat[val=0]{0}
    \eventat[val=2]{1}
    \eventat[val=1]{5}
  \end{stream}
  \filldraw (0,-8) node[left, anchor=east, every label] {ignor.}
    --++(2,0) --++(0,.7) --++(1,0) --++(0,.7) --++(1,0) --++(0,-.7) --++(1,0) --++(0,-.7) --++(1,0) -- cycle;
\end{streampicture}
\end{wrapfigure}

To compare $\gamma(s)$ and $S$ in a quantitative way we define the
\emph{ignorance measure}
$\iota: 2^{{\mathcal S}_\mathbb D} \to \mathbb I = [0,1]$ scoring the
ambiguity of such a set of streams, i.e., how similar the different
streams in the set are.
Events in non-synchronized streams might not have corresponding events
at the same timestamp on the other streams.
Hence we refer to the signal semantics of event streams where the
events represent the changes of a piece-wise constant signal.
As depicted on the right with three event streams over the finite data
domain $\{0,1,2\}$, we score timestamps based on how many event
streams have the same value with respect to the signal semantics at
that timestamp.
These scores are then integrated and normalized throughout the length
of the streams.
See Appendix~\ref{app:ignorance-measure} for the technical details.
Using this ignorance measure we can now compute the optimal ignorance
$i := \iota(S)\in \mathbb I$ and the ignorance
$k := \iota(\gamma(s))\in \mathbb I$ of the streams produced by the
abstract TeSSLa specification.
%

For the evaluation we took several example specifications and
corresponding input traces representing different use-cases of TeSSLa
and compared the optimal ignorance with the ignorance of abstract
TeSSLa.
Note that computing the optimal ignorance requires to derive all
possible variants of events that might have happened during gaps,
which are in general infinitely many and in the special case of only
point-wise gaps still exponentially many.
Hence this can only be done on rather short traces with only a few
point-wise gaps.
As a measure for the overhead imposed by using the abstraction
compared to the concrete TeSSLa specification we use the computation
depth, i.e., the depth of the dependency graph of the computation
nodes of the specifications.
While runtimes are highly depending on implementation details of the
used TeSSLa engines, the computation depth is a good indicator for the
computational overhead in terms of how many concrete TeSSLa operators
are needed to realize the abstract TeSSLa specification.
See Fig.~\ref{fig:empirical} for a graphical representation
and Appendix~\ref{sec:empirical-data} for numerical results.

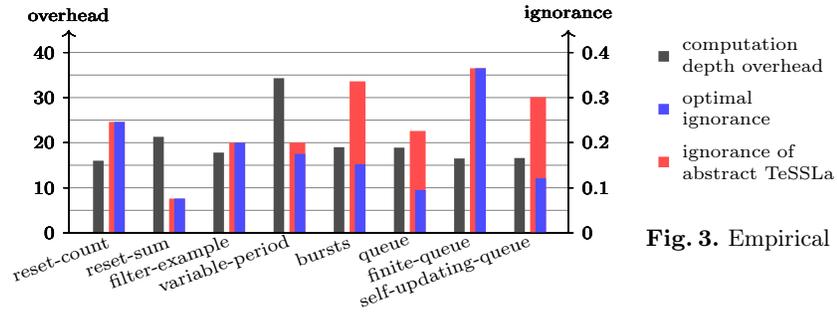
\begin{figure}[t]
  \centering
  \begin{tikzpicture}[y=.6mm,x=8mm,every node/.style={font=\scriptsize}]
    \foreach \y in {0,5,...,40} {
      \draw[very thin, gray] (-.8,\y) -- (7.5,\y);
    }
    \foreach \x/\name/\o/\i/\k in {
      0/reset-count/16.0/24.6/24.6,
      1/reset-sum/21.3/7.6/7.6,
      2/filter-example/17.8/20/20,
      3/variable-period/34.3/17.5/20,
      4/bursts/19.0/15.2/33.6,
      5/queue/18.9/9.5/22.6,
      6/finite-queue/16.5/36.5/36.5,
      7/self-updating-queue/16.6/12.1/30.1
    } {
      \node[below,anchor=east,rotate=20] at (\x,-2pt) {\name};
      \draw[line width=4pt, black!70] (\x,0) ++ (-7pt,0) -- ++(0,\o);
      \draw[line width=6pt, red!70] (\x,0) ++ (0,0) -- ++(0,\k);
      \draw[line width=4pt, blue!70] (\x,0) ++ (1pt,0) -- ++(0,\i);
      \draw (-.8,0) -- (7.5,0);
      \draw[->] (-.8,0) -- ++(0,45) node [above] {overhead};
      \foreach \o in {0,10,20,30,40} {
        \draw (-.8,\o) -- ++ (-2pt,0) node[left] {\o};
      }
      \foreach \i in {0,0.1,0.2,0.3,0.4} {
        \draw (7.5,\i*100) -- ++ (2pt,0) node[right] {\i};
      }
      \draw[->] (7.5,0) -- ++(0,45) node [above] {ignorance};
    }
    \draw[line width=4pt, draw=black!70] (9,45) ++ (0,-10pt) -- ++(4pt,0)
      node [right, text width=3cm, align=left] {computation\\ depth overhead};
    \draw[line width=4pt, draw=blue!70] (9,45) ++ (0,-30pt) -- ++(4pt,0)
      node [right, text width=3cm, align=left] {optimal\\ ignorance};
    \draw[line width=4pt, draw=red!70] (9,45) ++ (0,-50pt) -- ++(4pt,0)
      node [right, text width=3cm, align=left] {ignorance of\\ abstract TeSSLa};
  \end{tikzpicture}
  \vskip-1.6cm
  \hfill
  \parbox{3.5cm}{\caption{Empirical results.}}
  \label{fig:empirical}
  \vskip9mm
\end{figure}

The first three examples represent the class of common, simple TeSSLa
specifications without complex interdependencies and no generation of
additional events with \DELAY: Reset-count counts between reset
events; reset-sum sums up events between reset events; and
filter-example filters events occurring in a certain timing-pattern.
For these common specifications the overhead is small and the
abstraction is perfectly accurate.
The burst example checks if events appear according to a complex
pattern.
In the abstraction we loose accuracy because the starting point of a
burst is not accessible by \LASTabs after a gap.
A similar problem occurs in the queue example where we use a complex
data domain to develop a queue along an event stream.
If \LASTabs produces $\top$ after a gap all information about the
queue before the gap is lost.
For variable-period the abstraction is not perfectly accurate, because
the \DELAY is used to generate events periodically depending on an
external input.
This gets even worse for the self-updating queue where complex
computations are performed depending on events generated by a \DELAY.
Surprisingly, the finite-queue is again perfectly accurate, because
the size of the queue is limited in a way that eliminates the
inaccuracy of the abstraction in this particular example.


%% file: conclusion.tex

\section{Conclusion}
\label{sec:conclusion}
By replacing the basic operators of TeSSLa with abstract counterparts, we obtained a framework where properties and analyses can be specified with respect to complete traces and automatically evaluated for partially known traces.
We have shown that these abstract operators can be encoded in TeSSLa, allowing existing evaluation engines to be reused.
This is particularly useful as TeSSLa comprises a very small core language suitable for implementation in soft- as well as hardware.
Using the example of sliding windows, we demonstrated how complex data structures like queues can be abstracted.
Using finite abstractions, our approach even facilitates using complex data structures when only limited memory is available.
Evaluating the abstract specification typically only increases the computational cost by a constant factor.
In particular, if a concrete specification can be monitored in linear time (in the size of the trace) its abstract counterpart can be as well.
Finally, we illustrated the practical feasibility of our approach by an empirical evaluation using the freely available
\tessla engine.

%% file: appendix.tex

\newpage
\appendix

\raggedbottom

\section{TeSSLa Evaluation Example} \label{app:semantics-example}

Consider the equation
\[ y = \KW{merge}\big(\KW{lift}\big({} + 1\big)\big(\KW{last}(y, x)\big), 0\big) \]
where $0 = \DEF{zero} = \DEF{const}(0, \DEF{unit}) = \DEF{const}(0, \unitsym\infty) = 0\infty$.

The equation can be seen as function of $y$ with fixed input stream $x = 2\unitsym4\unitsym\infty$. To compute the least fixed-point of this function we start with the empty stream $y_0 = 0$. We than have
\begin{align*}
  \KW{last}(y, x) &= 0\bot \\
  \KW{lift}\big({} + 1\big)\big(\KW{last}(y, x)\big) &= 0\bot \\
  \KW{merge}\big(\KW{lift}\big({} + 1\big)\big(\KW{last}(y, x)\big), 0\big) &= 00
\end{align*}
For the next iteration we start with $y_1 = 00$:
\begin{align*}
  \KW{last}(y, x) &= 2 \\
  \KW{lift}\big({} + 1\big)\big(\KW{last}(y, x)\big) &= 2 \\
  \KW{merge}\big(\KW{lift}\big({} + 1\big)\big(\KW{last}(y, x)\big), 0\big) &= 002
\end{align*}
For the next iteration we start with $y_2 = 002$:
\begin{align*}
  \KW{last}(y, x) &= 20 \\
  \KW{lift}\big({} + 1\big)\big(\KW{last}(y, x)\big) &= 21 \\
  \KW{merge}\big(\KW{lift}\big({} + 1\big)\big(\KW{last}(y, x)\big), 0\big) &= 0021
\end{align*}
For the next iteration we start with $y_3 = 0021$:
\begin{align*}
  \KW{last}(y, x) &= 204 \\
  \KW{lift}\big({} + 1\big)\big(\KW{last}(y, x)\big) &= 214 \\
  \KW{merge}\big(\KW{lift}\big({} + 1\big)\big(\KW{last}(y, x)\big), 0\big) &= 00214
\end{align*}
For the next iteration we start with $y_4 = 00214$:
\begin{align*}
  \KW{last}(y, x) &= 2041\infty \\
  \KW{lift}\big({} + 1\big)\big(\KW{last}(y, x)\big) &= 2142\infty \\
  \KW{merge}\big(\KW{lift}\big({} + 1\big)\big(\KW{last}(y, x)\big), 0\big) &= 002142\infty
\end{align*}
For the next iteration we start with $y_5 = 002142\infty$:
\begin{align*}
  \KW{last}(y, x) &= 2041\infty \\
  \KW{lift}\big({} + 1\big)\big(\KW{last}(y, x)\big) &= 2142\infty \\
  \KW{merge}\big(\KW{lift}\big({} + 1\big)\big(\KW{last}(y, x)\big), 0\big) &= 002142\infty
\end{align*}
So we have reached a fixed-point.

Note that $y_0 \sqsubseteq y_1 \sqsubseteq y_2 \sqsubseteq y_3 \sqsubseteq y_4 \sqsubseteq y_5$ regarding the prefix relation.

\section{Delay} \label{app:delay}

In Section~\ref{sec:tessla} TeSSLa was introduced with the operators $\KW{nil}$, $\KW{unit}$,
$\KW{time}$, $\KW{lift}$ and $\KW{last}$. As discussed in \cite{tessla} this set of operators
is sufficient to express timestamp conservative functions, i.e., functions which do not introduce
additional timestamps. In order to express arbitrary functions which can add events at arbitrary
timestamps, TeSSLa was introduced in \cite{tessla} with the \KW{delay} operator.
The \KW{delay} operator was removed
from the main content of this paper due to space limitations and the amount of technical details needed
to define the concrete and abstract operators. Since these two operators nevertheless
fit very seamless into the abstraction framework presented in this paper they are included
in the implementation and the empirical evaluation and their definitions are given in this
appendix.

\subsection{Concrete Delay}

\KWbullet
$\KW{delay}: 𝓢_{𝕋∖｛0｝}×𝓢_𝔻 → 𝓢_𝕌, \KW{delay}(d, r) := z$
takes a delay stream $d$ and a reset stream $r$.
It emits a unit event in the resulting stream after the delay passes.
Every event on the reset stream resets any delay.
New delays can only be set together with a reset event or an emitted output event.
Formally,
\[
z(t) =
\begin{cases}
  \unitsym & \exists_{t'<t} d(t') = t - t' \wedge \AUX{setable}(z,r,t') \wedge \AUX{noreset}(r,t', t) \\
    ⊥ & \AUX{defined}(z,t) \wedge \forall_{t'<t} d(t') ≠ t - t' \wedge d(t') ≠ \mathrm{?} \\
    ⊥ & \forall_{t'<t}\AUX{unsetable}(z,r,t') \vee \AUX{reset}(r,t', t)\\
    ? & \text{otherwise}
\end{cases}
\]
with the following auxiliary functions:
\begin{align*}
  \AUX{setable}(z,r,t') &\DefinedAs z(t')=\unitsym ∨ t'\in\Ticks(r), \\
  \AUX{reset}(r,t,t') &\DefinedAs ∃_{t'' | t < t'' < t'} t''\in\Ticks(r), \\
  \AUX{unsetable}(z,r,t') &\DefinedAs z(t')=\bot ∧  r(t') = \bot, \text{ and} \\
  \AUX{noreset}(r,t,t') &\DefinedAs ∀_{t'' | t < t'' < t'} r(t'')=\bot.
\end{align*}

For an example on how to use the \KW{delay} operator in a TeSSLa specification see Section~\ref{sec-self-updating-queue}.

\subsection{Abstract Delay}

In Section~\ref{sec:abstraction} of this paper we presented the abstract operators \KWabs{lift} and \KWabs{last} for
their concrete counterparts \KW{lift} and \KW{last}. In a similar way we now define the abstract \KWabs{delay} operator
as counterpart of the concrete \KW{delay}.

\KWabsbullet
$\DELAYabs: \mathcal{P}_{\mathbb{T}_\infty \setminus \{0\}}
  \times \mathcal{P}_\mathbb{D} \rightarrow \mathcal{P}_\mathbb{U}, \DELAYabs(d, r) := z$
needs additional logic to handle $\top$ events on the delay stream and gaps on the input streams. Since the output stream is of type $\mathcal P_{\mathbb U}$ no additional logic to produce $\top$ events is needed. We introduce the following additional auxiliary functions:
\begin{compactenum}[(1)]
\item $\AUX{maybeEvent}$, which captures whether there may be an
  event that would cause the delay to hold:
\[
\begin{array}{rl}
  \AUX{maybeSetable}(d,r,t) \DefinedAs & z(t) \in \{\boxempty,\gap\} \lor r(t) \in \mathbb{D}^\# \cup \{\gap\}  \\
  \AUX{maybeNoReset}(r,t, t') \DefinedAs & \forall_{t'' | t < t'' < t'} t'' \notin \Ticks(r) \\
  \AUX{maybeEvent}(d,r,t) \DefinedAs & \exists_{t' < t}d(t') \in \{t-t', \top, \gap\} \land \AUX{maybeSetable}(d,r,t') \land {}\\
                          & \AUX{maybeNoReset}(r,t',t)
\end{array}
\]    
\item $\AUX{maybeBot}$, which captures whether it is plausible that
  the stream $z$ being defined does not tick:
\[
\begin{array}{rl}
  \AUX{maybeUnsetable}(d,r,t) \DefinedAs& t \notin \Ticks(z) \land t \notin \Ticks(r)\\
  \AUX{maybeReset}(r,t, t') \DefinedAs& \exists_{t'' | t < t'' < t'} r(t'') \in \mathbb{D}^\#\cup\{\gap\}\\
  \AUX{maybeBot}(d,r,t) \DefinedAs & \AUX{defined}(z,t) \land \forall_{t' < t}d(t') \not= \mathrm{?} \land d(t') \not= t-t' \lor {} \\
  &\AUX{maybeUnsetable}(d,r,t') \lor \AUX{maybeReset}(r,t',t)
\end{array}
\]
\end{compactenum}
The conjunction of these two predicates captures the gap case. Again, the parts similar to the concrete operator are typeset in gray:
\[
z(t) = \begin{cases}
    \unitsym & \color{gray} \exists_{t' < t} d(t') = t - t' \land \AUX{setable}(d,r,t') \land \AUX{noreset}(r,t', t) \\
    ⊥ & \color{gray} \AUX{defined}(z,t) \land \forall_{t'<t} d(t') ≠ t - t' \land \color{black} d(t')\not\in\{\textrm{?},\gap\} \\
    ⊥ & \color{gray} \forall_{t'<t} \AUX{unsetable}(z,r,t') \lor \AUX{reset}(r,t', t)\\
    \gap & \AUX{maybeEvent}(d,r,t) \land \AUX{maybeBot}(d,r,t) \\
    \mathrm{?} & \text{otherwise}
  \end{cases}
\]

\begin{wrapfigure}[5]{r}{70mm}
\vskip-8mm
\raggedleft
\begin{streampicture}[gap/.append style={fill=regioncolor}, xscale=0.75, y=8mm]
  \begin{stream}{d}
    \definedto{13}
    \undefinedto{16}
    \definedto{17}
    \delayeventat{1}{2}
    \gapat{2}
    \delayeventat{4}{2}
    \gapat{5}
    \delayeventat{7}{1}
    \delayeventat{9}{2}
    \topat{12}
  \end{stream}
  \begin{stream}{r}
    \definedto{17}
    \eventat{4}
    \gapat{7}
    \eventat{9}
    \gapat{10}
    \eventat{12}
    \eventat{13}
    \gapat{14}
    \eventat{15}
    \eventat{16}        
  \end{stream}
  \begin{stream}{z}
    \definedto{2}
    \definedto{12}
    \undefinedto{13}
    \definedto{14}
    \undefinedto{16}
    \definedto{17}
    \gapat{16}
    \gapat{13}
    \eventat{6}
    \gapat{8}
    \gapat{11}
  \end{stream}
  \streamregion{1}{3}\label{reg:gap}
  \streamregion{4}{6}\label{reg:topgap}
  \streamregion{7}{8}\label{reg:reset}
  \streamregion{9}{11}\label{reg:top}
  \streamregion{12}{13}\label{reg:nogap1}
  \streamregion{14}{16}\label{reg:nogap2}
\end{streampicture}
\end{wrapfigure}

The trace diagram on the right shows an example covering most of the interesting edge cases of the abstract delay:
\ref{reg:gap}, \ref{reg:topgap} the gap on $d$ has no effect;
\ref{reg:reset} the gap on $r$ makes it unclear if a delay starts, a gap is created on $z$;
\ref{reg:top} the gap on $r$ makes it unclear if reset happens, a gap is created on $z$.
\ref{reg:nogap1} $\top$ causes unknown delay which results in a gap on $z$ until next event happens on $r$;
\ref{reg:nogap2} the gap on $r$ starts unknown delay because of gap on $d$, reset event on $r$ does not help because gap on $d$ remains. Results in long gap on $z$ until event on $r$ happens out of the gap on $d$.

\subsection{Finite State Implementations of Delay}

\begin{wrapfigure}[3]{r}{35mm}
\vskip-8mm\raggedleft
\begin{streampicture}[xscale=.59,y=7mm]
  \begin{stream}{d}
    \definedto{10}
    \delayeventat{1}{2}
    \delayeventat{4}{3}
    \delayeventat{5}{3}
    \delayeventat{6}{3}
  \end{stream}
  \begin{stream}{r}
    \definedto{2}
    \eventat{1}
    \undefinedto{8}
    \definedto{10}      
  \end{stream}
  \begin{stream}{z}
    \definedto{10}
    \draw[red] (7,-6) -- +(2,0);
    \gapat{3}
    \gapat[draw=red]{7}
    \gapat[draw=red]{8}
    \gapat[draw=red]{9}
  \end{stream}
\end{streampicture}
\end{wrapfigure}

Consider the example for $z = \DELAYabs(d,r)$ shown on the right.
If the delay produces an event after or during a gap on the reset
stream, a potential reset event might cancel the delay earlier,
creating point-wise gaps.
The concrete delay has only one active delay at any given point in time.
However, to implement the abstract delay infinite memory may be needed
to keep track of all potential delays, because it may not be known
whether a delay has been canceled or started earlier by a reset event.
By keeping track only of the minimal and maximal potential delays we
get a finite-memory implementation. 
However, this implementation may be an imperfect abstraction of the
delay operator which produces larger gaps.
Those gaps marked in red on the right are replaced with one large gap.

\begin{definition}[Finite Memory Abstract Delay]
  We define
  $\DELAYabs_{\text{fin}}(d,r) = z$ to be the same as $\DELAYabs(d,r)$
  except that $z(t) := \gap$ (instead of $\bot$) if
  $\exists_{t' < t} t' + d(t') \leq t \land r(t') = \gap$ and $\forall_{t''|t' + d(t') \leq t'' < t } z(t'') = \gap$ and $\nexists_{t''|t' < t'' < t} t'' \in \Ticks(r)$ and $\exists_{t' < t} t' + d(t') \geq t \land r(t') = \gap$.
\end{definition}

\subsection{Unrolling Delay}

In Section~\ref{sec:well-formedness} we discussed that using the abstract
TeSSLa operators might introduce additional cyclic dependencies between
the stream variables. In order to overcome this issue we presented
translation in Definition~\ref{def:unrolling} which reintroduces guards
in the cyclic dependency by unrolling the value and gap calculation
sequentially. The same approach can be applied to unroll
the abstract \KWabs{delay} operator:

\begin{definition}[Unrolled Abstract Delay]
  \label{def:unrolled-delay}
    We define two variants of the abstract delay, $\DELAYbot$ and $\DELAYgap$ as follows:
    Let $z = \DELAYabs(d,r)$, then $\DELAYbot(d,r) := z_\bot$ and
    $\DELAYgap(d,r,p) := z_{\gap}$.
    \[z_\bot(t) = \begin{cases}
    z(t) & \text{if } z(t) \not= \gap \\
    \bot & \text{otherwise}
    \end{cases}
    \qquad
    z_{\gap}(t) = \begin{cases}
    p(t) & \text{if } t \in \Ticks(p) \\
    \gap & \text{if } t \notin \Ticks(p) \land z(t) = \gap \\
    \bot & \text{otherwise}
    \end{cases}\]
\end{definition}

Since the trigger input of a delay operator cannot be recursive in a
well-formed specification, a recursive equation using one last has
the form $x = \DELAYabs(d,r)$ and $d = f(x,\vec c)$,
where $\vec c$ is a vector of streams not involved in the
recursion and $f$ does not introduce further lasts or delays. 
Now, this equation system can be rewritten in the following equivalent form:
\[
x' = \DELAYbot(d,r)
\qquad
d' = f(x',\vec c)
\qquad
x = \DELAYgap(d',r,x')
\qquad
d = f(x,\vec c)
\]
By doing this, the calculation of the value and the domain of the
definition is split in two parts preserving the semantics. The same pattern can be repeated if multiple recursive abstract lasts or delays
are used.

There is a remaining case to be considered, when $f$ contains
other abstract lasts or delays which lead to mutual recursive specifications.
In this case, these other abstract lasts and delays have to use $x$,
$d$ or $r$ as their input instead of the primed counterparts because
the computation requires the complete last or delay stream, respectively.
Furthermore, there may be more unrolling steps needed than just one to unroll the whole specification but every unrolling step still follows the previously described pattern for last and delay.

\section{Missing Proofs} \label{app:proof}

\setcounter{theorem}{1}
\begin{theorem}
	Every abstract TeSSLa operator is a perfect abstraction of its concrete counterpart.
\end{theorem}

\begin{proof}
	We need to show that for all partial streams $s_1,\dots,s_n$ with $s_i \in {\mathcal{P}_{\mathbb{D}_i}}_i$ and every TeSSLa operator $f$ and its abstract counterpart $f^\#$ it holds that $f^\#(s_1,\dots,s_n) = \alpha(f(\gamma(s_1),\dots,\gamma(s_n)))$. We use the Galois Connection $(\alpha,\gamma)$ as presented in Definition \ref{def:pes}. For $\NIL$ and $\UNIT$ this holds trivially because they are equal to their abstract counterparts. So it remains to show this for $\TIME$, $\LIFT$, $\LAST$ and $\DELAY$. Note that because the Galois Connection maps from the abstract lattice to the powerset of the (concrete) lattice, we assume all operators in the following are naturally extended to sets.
	
	$\TIMEabs(s_1) = \alpha(\TIME(\gamma(s_1)))$ holds because $\TIMEabs$ changes nothing regarding the abstract elements: gaps are copied and the values of the events do not matter because they are replaced with the timestamps.
	
	$\LIFTabs(g^\#)(s_1,\dots,s_n) = \alpha(\LIFT(g)(\gamma(s_1),\dots,\gamma(s_n)))$ holds if $g^\#$ is a perfect abstraction of $g$, because the values and gaps are only passed to $g^\#$ in the abstract case which calculates the values and gaps.
	
	$\LASTabs(s_1,s_2) = \alpha(\LAST(\gamma(s_1),\gamma(s_2)))$ holds because 
	\begin{itemize}
		\item the only special case regarding the data abstraction is when a trigger occurs and there is a gap on the stream of values after the last event. Then $\top$ is the output which is fine because in the concretization the gap is replaced with arbitrary values. Other then that, the data abstraction is not important because in $\LASTabs$ the data values are only copied to the new events on the output stream.
		\item the only special case regarding the gaps is a trigger event before any event on the stream of values but after a gap on said stream, which results in a gap. This is fine because the gap on the stream of values is later replaced with an arbitrary or no value which is then the output when the trigger event occurs. Thus using $\alpha$ later on this set of streams, we get a gap where the trigger event was. After the first event on the stream of values, the gaps from the stream of trigger events are just copied.
	\end{itemize}
	
	$\DELAYabs(s_1,s_2) = \alpha(\DELAY(\gamma(s_1),\gamma(s_2)))$ holds because data abstraction wise, the output can always only be unit. If we look at the gaps, there are only a few cases where $\DELAYabs$ produces gaps: Point-gaps are created when a delay times out which (1) is set via a gap or (2) is set normally but a gap occurs on the reset stream before the timeout. Big gaps which last until the next event on the reset stream are created when a delay can be set but on the delay stream is (3) an event with $\top$ as value or (4) a gap. This perfectly fits to how $\alpha(\DELAY(\gamma(d),\gamma(r)))$ works because in
	\begin{itemize}
		\item[(1)] the concretization function creates sets of streams where the delay is set not at all or with a value. The $\DELAY$ then outputs a set of streams where sometimes there is an output event and sometimes not which results in a point-gap again when the abstraction function is used.
		\item[(2)] the concretization function creates sets of streams where where a reset and no reset happens before the timeout. This has the same implications as for (1).
		\item[(3)] there is a delay started with $\top$ as value. Then the concretization function creates a set of streams having all possible values as delay at that point. Hence, $\DELAY$ results in the set of streams that have an arbitrary amount of events between the point where the delay is set and the next reset event. Using the abstraction function again this results in a big gap.
		\item[(4)] it is the same case as (3) but the set of streams created by the concretization function contains additionally those streams with no value as delay. This has the same implications as in (3).
	\end{itemize}
	\qed
\end{proof}

\setcounter{theorem}{4}
\begin{lemma}
	$\LASTTIMEabs$ is a perfect abstraction of $\LASTTIME$.
\end{lemma}

\begin{proof}
	We need to show that for all partial streams $v,r$ the following equation holds: $\LASTTIMEabs(v,r) = \alpha(\LASTTIME(\gamma(v),\gamma(r)))$. We use the Galois Connection $(\alpha,\gamma)$ as presented in Definition \ref{def:pes}. Note that because the Galois Connection maps from the abstract lattice to the powerset of the (concrete) lattice, we assume all operators in the following are naturally extended to sets.
	
	Because as shown in Theorem \ref{theorem:abstraction} $\LASTabs$ and $\TIMEabs$ are perfects abstractions of $\LAST$ and $\TIME$, respectively, to show that $\LASTTIMEabs$ is a perfect abstraction of $\LASTTIME$, we just need to show that for the cases where the composition of $\LASTabs$ and $\TIMEabs$ is not a perfect abstraction of $\LASTTIME$, $\LASTTIMEabs$ is now perfect. More precisely, the only case which is not perfect in the composition is the $\top$ case, which means a gap occurred after the last event on the stream of values. Then $\LASTTIMEabs$ returns an interval from the timestamp of the last event on $v$ to the timestamp at the end of the gap. This is the same which would be done when the concretization function first creates a set of streams containing an arbitrary amount of events in the gap which results in different timestamps as output. Either the latest from an event in the gap, if one exists or from the last event on $v$. The abstraction function would then creat said interval. In every other case, $\LASTTIMEabs$ outputs the same value as the composition of the perfect operators does.
	\qed
\end{proof}

\section{Data Abstraction of a Queue for Sliding Window Average Computation}
\label{sec:appqueue}

\subsection{Concrete Queue}

TeSSLa can handle complex data structures and corresponding functions.
An ordered queue is a sequence of timestamps and values
$\mathcal{Q}_\mathbb{D} = (\mathbb{T} \times \mathbb{D})^*$
such that
\[\forall \langle t_0, d_0, t_1, d_1, \ldots, t_n, d_n \rangle \in \mathcal{Q}_\mathbb{D}: t_{i-1} < t_i \text{ where } 1 \leq i \leq n\]
We use the following functions on queues:
\begin{align*}
\AUX{enq}&: \mathbb{T} \times \mathbb{D} \times \mathcal{Q}_\mathbb{D} \rightarrow \mathcal{Q}_\mathbb{D}\\
\AUX{enq}(t, d, q) &= q \langle t,d \rangle \\
\AUX{remOlder}_k&: \mathbb{T} \times \mathcal{Q}_\mathbb{D} \rightarrow \mathcal{Q}_\mathbb{D}\\
\AUX{remOlder}_k(t, \langle t_1,d,t_2 \rangle q) &= \begin{cases}
\AUX{remOlder}_k(t, \langle t_2 \rangle q) & \text{if } t_2 \leq t-k\\
\langle \AUX{max}(t-k,t_1), d, t_2\rangle q & \text{else}\\	
\end{cases}\\
\AUX{remOlder}_k(t, \langle t_1,d \rangle) &= \begin{cases}
\langle t-k,d\rangle & \text{if } t_1 < t-k\\
\langle t_1,d \rangle & \text{else} 
\end{cases}\\
\AUX{remOlder}_k(t, \langle\rangle) &= \langle \rangle\\
\AUX{fold}&: ((\mathbb{T}^2 \times \mathbb{D}) \times \mathbb{D} \rightarrow \mathbb{D}) \times \mathcal{Q}_\mathbb{D} \times \mathbb{D} \times \mathbb{T} \rightarrow \mathbb{D}\\
\AUX{fold}(f, \langle t_1,d,t_2\rangle q, acc, until) &= \AUX{fold}(f,\langle t_2\rangle q, f(t_1,t_2,d,acc), until)\\
\AUX{fold}(f, \langle t_1,d \rangle, acc, until) &= f(t_1,until,d,acc)\\
\AUX{fold}(f, \langle \rangle, acc, until) &= acc\\
\end{align*}

\subsection{Abstract Queue}

The set of our abstract version of queues we use here is defined as $\mathcal{Q}^\#_{\mathbb{D}^\#} = \mathbb{T}_\infty \times \mathcal{Q}_{\mathbb{D}^\#}$. The concretisation function is given as follows
\begin{align*}
\gamma((t, q)) = & \{\langle(t_0,d_0)\dots(t_n,d_n)(t_{n+1},d_{n+1})\dots(t_m,d_m)\rangle \mid \\ 
& \langle(t_0,d_0)\dots(t_n,d_n)\rangle \in \mathcal{Q}_{\mathbb{D}^\#} \land t_n < t \land \\
& \langle(t_{n+1},d_{n+1})\dots(t_m,d_m)\rangle \subseteq q \land t_{n+1} \geq t\}
\end{align*}
and the abstraction functions is given as $\alpha(Q) = \mathsf{sup}\{(0, q) \mid q \in Q\}$.

To improve the accuracy of the abstraction we insert the following limiting function as annotations on the possible ranges of variables into the concrete specification:
\begin{align*}
\AUX{limit}&: \mathbb{D} \times \mathbb{D} \times \mathbb{D} \rightarrow \mathbb{D}\\
\AUX{limit}(a,b,d) &= \begin{cases}
a & \text{if } d < a \\
b & \text{if } d > b \\
d & \text{else}
\end{cases}
\end{align*}
In a similar fashion we add the a function removing newer elements from the queue which serves as annotation that a queue will never contain elements with a future timestamp:
\begin{align*}
\AUX{remNewer}&: \mathbb{T} \times \mathcal{Q}_\mathbb{D} \rightarrow \mathcal{Q}_\mathbb{D}\\
\AUX{remNewer}(t, q \langle t_1,d \rangle) &= \begin{cases}
\AUX{remNewer}(t, q) & \text{if } t_1 > t\\
q \langle t_1, d\rangle & \text{else}\\ 
\end{cases}\\
\AUX{remNewer}(t, \langle\rangle) &= \langle\rangle
\end{align*}
Now we define abstract versions for the previously defined functions:
\begin{align*}
\AUX{remOlder}_k^\#&: \mathbb{T} \times \mathcal{Q}^\#_{\mathbb{I}^\#_\mathbb{D}} \rightarrow \mathcal{Q}^\#_{\mathbb{I}^\#_\mathbb{D}}\\
\AUX{remOlder}_k^\#(t, (u, q)) &= \begin{cases}
(u, \AUX{remOlder_k(t,q)}) & \text{if } t - k < u\\
(0, \AUX{remOlder_k(t,q)}) & \text{else}\\	
\end{cases}\\
\AUX{enq}^\#&: \mathbb{T} \times \mathbb{I}^\#_\mathbb{D} \times \mathcal{Q}^\#_{\mathbb{I}^\#_\mathbb{D}} \rightarrow \mathcal{Q}^\#_{\mathbb{I}^\#_\mathbb{D}}\\
\AUX{enq}^\#(t, d, (u,q)) &= (u, \AUX{enq}(t, d, q)) \\
\AUX{fold}^\#&: ((\mathbb{T}^2 \times \mathbb{I}^\#_\mathbb{D}) \times \mathbb{I}^\#_\mathbb{D} \rightarrow \mathbb{I}^\#_\mathbb{D}) \times \mathcal{Q}^\#_{\mathbb{I}^\#_\mathbb{D}} \times \mathbb{I}^\#_\mathbb{D} \rightarrow \mathbb{I}_\mathbb{D}\\
\AUX{fold}^\#(f, (u, q), acc, until) &= \begin{cases}
\AUX{fold}(f, q, acc, until) & \text{if } u = 0\\
\AUX{fold}(f, q, \top, until) & \text{else}
\end{cases}\\
\AUX{remNewer}^\#&: \mathbb{T} \times \mathcal{Q}^\#_{\mathbb{I}^\#_\mathbb{D}} \rightarrow \mathcal{Q}^\#_{\mathbb{I}^\#_\mathbb{D}}\\
\AUX{remNewer}^\#(t, (u, q)) &= (\AUX{min}(u,t), \AUX{remNewer}(t,q)) \\
\AUX{limit}^\#&: \mathbb{D} \times \mathbb{D} \times \mathbb{I}^\#_\mathbb{D} \rightarrow \mathbb{I}^\#_\mathbb{D}\\
\AUX{limit}^\#(a,b,[d,d']) &= [\AUX{limit}(a,b,d),\AUX{limit}(a,b,d')]
\end{align*}

Note, that the abstract functions use intervals $\mathbb I^\#_{\mathbb D}$ to represent ranges of possible values.

In order to use the data domains we have to name the following elements:
\[ \top = (\infty, \langle\rangle ) \in \mathcal{Q}^\#_{\mathbb{I}^\#_\mathbb{D}}\text{ and }\top = [-\infty,\infty] \in \mathbb{I}^\#_\mathbb{D}. \]
Now we can provide an abstraction of the TeSSLa specification for the input stream $\DEF{load}\in \mathcal{P}_{\mathbb{R}^\#}$:
\begin{align*}
  \DEF{stripped} &= \DEF{slift}^\#(\AUX{rem})(\TIMEabs(\DEF{load}), \DEF{merge}^\#(\LASTabs(\DEF{queue}, \DEF{load}), \langle\rangle)))\\
  \AUX{rem}(t,q) &= \AUX{remOlder}^\#_5(t,\AUX{remNewer}^\#(t,q))\\
  \DEF{queue} &= \LIFTabs(\AUX{enq}^\#)(\TIMEabs(\DEF{load}), \DEF{load}, \DEF{stripped})\\
  \DEF{avg} &= \LIFTabs(\AUX{int})(\DEF{queue}, \TIMEabs(\DEF{load}))\\
  \AUX{int}(q, u) &= \AUX{fold}^\#(f_u, q, 0, u) \\
  f_u(a, b, v, \AUX{acc}) &= \AUX{limit}^\#\left(0, \frac{a - u + 5}5, \AUX{acc}\right) + v \cdot (b-a)/5
\end{align*}

\subsection{Finite Abstract Queue}

If we want to limit the queue to a certain number of events, we just need to change the enqueue function as follows:
\begin{align*}
  \AUX{enq}_c&: \mathbb{T} \times \mathbb{I}^\#_\mathbb{D} \times \mathcal{Q}^\#_{\mathbb{I}^\#_\mathbb{D}} \times \mathbb{N}_{>1} \rightarrow \mathcal{Q}^\#_{\mathbb{I}^\#_\mathbb{D}}\\
  \AUX{enq}_c(t, d, (u, q), n) &= \begin{cases}
    \AUX{enq}^\#(t, d, (u, q)) & \text{if } |q| < n \\
    \AUX{enq}_c(t, d, (t_2, \langle t_2 \rangle q'), n) & \text{otherwise,}
  \end{cases}\\
  & \hspace{5em} \text{with } q = \langle t_1, d_1, t_2 \rangle q'
\end{align*}
For $q \in \mathcal{Q}^\#_{\mathbb{I}^\#_\mathbb{D}}$ we denote with $|q|$ the number of data values in $q$.

\subsection{Concrete Self-Updating Queue}
\label{sec-self-updating-queue}

We now update the specification to be self-updating, i.e., update automatically every time an event leaves the moving window completely. For this, we can use the operator \DELAY to create additional events. We add the following equations to the original specification and use \DEF{updated} instead of \DEF{queue} and \DEF{load} in the computation of \DEF{avg}:
\begin{align*}
\DEF{timeout} &= \LIFT(t, q \mapsto \AUX{dataTimeout}(q) - t + 5)(\TIME(\DEF{updated}), \DEF{updated})\\
\DEF{merged} &= \DEF{merge}(\DEF{queue}, \LAST(\DEF{updated}, \DELAY(\DEF{timeout}, \DEF{load})))\\
\DEF{updated} &= \LIFT(\AUX{remOlder}_5)(\TIME(\DEF{merged}), \DEF{merged})
\end{align*}
where $\AUX{dataTimeout}: \mathcal{Q}_\mathbb{D} \to \mathbb{T}$  is defined as follows:
\begin{align*}
\AUX{dataTimeout}(\langle\rangle) &= \infty\\
\AUX{dataTimeout}(\langle t,d\rangle) &= \infty\\
\AUX{dataTimeout}(\langle t,d,t'\rangle q) &= t'
\end{align*}

\subsection{Abstract Self-Updating Queue}
\label{sec-abstract-self-updating-queue}

The abstract version of the specification for the self-updating queue can again be derived by replacing all TeSSLa operators with their abstract counterparts:
\begin{align*}
\DEF{timeout} &= \LIFTabs(t, q \mapsto \AUX{dataTimeout}^\#(q) - t + 5)(\TIMEabs(\DEF{updated}), \DEF{updated})\\
\DEF{merged} &= \DEF{merge}^\#(\DEF{queue}, \LASTabs(\DEF{updated}, \DELAYabs(\DEF{timeout}, \DEF{load})))\\
\DEF{updated} &= \LIFTabs(\AUX{remOlder}_5)(\TIMEabs(\DEF{merged})), \DEF{merged})
\end{align*}
where $\AUX{dataTimeout}^\#: \mathcal{Q}^\#_{\mathbb{I}^\#_{\mathbb D}} \to \mathbb{T}$  is defined as follows:
\begin{align*}
\AUX{dataTimeout}^\#(u, q) &= \begin{cases}
\AUX{dataTimeout}(q) & \text{if } u = 0 \\
\top & \text{else}
\end{cases}
\end{align*}
The rest of the specification is adjusted to the abstract setting as shown before.
The unrolling of \LASTabs and \DELAYabs as defined in \autoref{def:unrolling} and \autoref{def:unrolled-delay} can be applied in order to realize this abstract TeSSLa specification in TeSSLa.

\section{Technical Details of the Ignorance Measure}
\label{app:ignorance-measure}

The function $v: \mathcal S_\mathbb D \times \mathbb T \to \mathbb D$ provides the current value of a stream at a given timestamp according to the signal semantics:
\begin{align*}
  v_s(t) &:= s(\max\{ t' \in \Ticks(s) \mid t' < t \})
\end{align*}

For the data domain $\mathbb D$ we assume a measure space $(\mathbb D, \mathcal A, \mu)$ with the $\sigma$-algebra $\mathcal A \subseteq 2^{\mathbb D}$ on $2^{\mathbb D}$ and the normalized measure $\mu: \mathcal A \to [0,1]$ on $\mathbb D$, s.t. $\forall D \in \mathcal A: 0 \le \mu(D) \le 1$ and $\mu(\mathbb D) = 1$.

An \emph{ignorance representation} is a piece-wise constant function $\ig: \mathbb T \to \mathcal A$ s.t. $\ig$ can be represented as a finite number of intervals. The progress of a stream $s \in \mathcal S_{\mathbb D}$ is a timestamp $t$ such that $\forall_{t' > t} s(t') = \mathrm{?}$ and $\forall_{t'<t} s(t') \neq \mathrm{?}$. For a set of streams $S \subseteq S_{\mathbb D}$ with equal progress $T$ an ignorance representation $\ig$ can be retrieved such that
\begin{align*}
  & \forall r,s \in S, t \in [0,T]: v_r(t) \neq v_s(t) \Rightarrow v_r(t), v_s(t) \in \ig(t) \text{ and}\\
  & \forall t \in \op{dom}(\ig): \forall x \in \ig(t): \exists r,s \in S: v_r(t) \neq v_s(t) = x.
\end{align*}

Note, that depending on the chosen $\sigma$-algebra $\mathcal A$ one might need to extend the set $S$ into a compliant superset by adding missing streams first, before a matching ignorance representation can be derived. In those cases one has to specify a corresponding extension function together with the $\sigma$-algebra $\mathcal A$ and measure $\mu$. We call the superset derived by applying the extension function to $S$ the \emph{ignorance closure} of $S$. From now on we assume $S$ always to be extended appropriately if needed. Under these assumptions, the tuple $(\mathcal S_\mathbb D, 2^{\mathcal S_\mathbb D})$ of all streams with progress $T$ and their power set form a measurable space.

The \emph{ignorance measure} $\iota: 2^{\hat {\mathcal S}_\mathbb D} \to [0,1]$ is defined via the interval representation of its ignorance representation:
\begin{align*}
  \iota(S) &:= \left.\left(\sum_{[t_1,t_2] \in \op{dom}(\ig)} (t_2 - t_1) \cdot \sum_{I \in \ig(t_1)}\mu(I)\right)\right/T
\end{align*}

For the evaluation performed in this paper we either use limited intervals or finite sets. In case of the limited intervals our data domain $\mathbb D$ is limited such that $\forall_{d\in\mathbb D}$ we have $k \le d \le l$. We use the measure space $(\mathbb D, \mathbb I_{\mathbb D}, \mu)$ using the set of all intervals $\mathbb I_{\mathbb D}$ on $\mathbb D$ such that for any interval $[a,b] \subseteq \mathbb D$ we have
\[ \mu([a,b]) = \frac{b-a}{l-k}.\]
In the case of a finite set $\mathbb D$ we can use the measure space $(\mathbb D, 2^{\mathbb D}, \mu)$ such that for any $D \subseteq \mathbb D$ we have
\[ \mu(D) = \frac{|D|}{|\mathbb D|}. \]

\section{Empirical Results}
\label{sec:empirical-data}

\begin{center}
  \begin{zebratabular}{lrrrrr}
    \headerrow
    example             & $d$ & $d^\#$ & $d^\#/d$ &  $i$ &  $k$ \\
    reset-count         &  4  &     64 &     16.0 & .246 & .246 \\
    reset-sum           &  4  &     85 &     21.3 & .076 & .076 \\
    filter-example      &  9  &    160 &     17.8 & .200 & .200 \\
    variable-period     &  3  &    103 &     34.3 & .175 & .200 \\
    bursts              & 13  &    247 &     19.0 & .152 & .336 \\
    queue               &  7  &    132 &     18.9 & .095 & .226 \\
    finite-queue        &  8  &    132 &     16.5 & .365 & .365 \\
    self-updating-queue & 11  &    183 &     16.6 & .121 & .301
  \end{zebratabular}
\end{center}

The above table shows the results of the empirical evaluation. The computation depth $d$ of the TeSSLa specification, the computation depth $d^\#$ of the corresponding abstract TeSSLa specification, the computation depth overhead $d^\#/d$, the perfect ignorance $i$ and the ignorance of the abstraction $k$.